\documentclass[12pt, svgnames]{article}
\usepackage{amsmath}
\usepackage{amssymb}
\usepackage{amsthm}
\usepackage{mathrsfs}
\usepackage{bbm}
\usepackage{tikz}
\usetikzlibrary{decorations.markings,decorations.pathreplacing,calc}
\usetikzlibrary{arrows.meta}
\usepackage{subcaption}
\usepackage{array}

\usepackage[framemethod=tikz]{mdframed}
\usepackage[margin=1in]{geometry}
\usepackage[utf8]{inputenc}
\usepackage{comment}
\usepackage{umoline}

\usepackage{hyperref}
\hypersetup{
     colorlinks=true,           
     linkcolor=purple,          
     citecolor=blue,            
     filecolor=blue,            
     urlcolor=cyan,             
 }

\theoremstyle{plain}
\newtheorem{theorem}{Theorem}[section]
\newtheorem{lemma}[theorem]{Lemma}

\newtheorem{claim}[theorem]{Claim}
\theoremstyle{definition}
\newtheorem{definition}[theorem]{Definition}

\newcommand {\br} [1] {\ensuremath{ \left( #1 \right) }}

\newcommand {\ket}[1] {\ensuremath{ \left| #1 \right\rangle }}

\DeclareMathOperator{\DL}{DL}
\DeclareMathOperator*{\Rob}{Rob}

\newcommand{\cRef}[1]{Ref.~\cite{#1}}

\newcommand {\suppress}[1]{}

\newcommand {\eps} {\varepsilon}

\newcommand{\SR}{\mathrm{SR}}

\newcommand{\prt}{\xi}

\newcommand{\EqDef}{\stackrel{\mathrm{def}}{=}}

\def\cG{\mathcal{G}}

\newcommand{\Eq}[1]{Eq.~(\ref{#1})}
\newcommand{\Fig}[1]{Fig.~\ref{#1}}
\newcommand{\Sec}[1]{Sec.~\ref{#1}}
\newcommand{\Thm}[1]{Theorem~\ref{#1}}
\newcommand{\Lem}[1]{Lemma~\ref{#1}}

\newcommand{\footremember}[2]{%
    \footnote{#2}
    \newcounter{#1}
    \setcounter{#1}{\value{footnote}}%
}

\newcommand{\ignore}[1]{}

\title{An area law for 2D frustration-free spin systems}
\author{%
  Anurag Anshu \footremember{ucb}{School of Engineering and Applied Sciences, Harvard University, Cambridge, USA} 
   \and Itai Arad\footremember{technion}{Physics Department, Technion, Israel.}%
  \and David Gosset \footremember{iqc}{Institute for Quantum Computing, University of Waterloo, Canada.} \footremember{co}{Department of Combinatorics and Optimization, University of Waterloo, Canada.}%
  }
\date{}
\begin{document} 
\maketitle
\begin{abstract}
  We prove that the entanglement entropy of the ground state of a
  locally gapped frustration-free 2D lattice spin system satisfies
  an area law with respect to a vertical bipartition of the lattice
  into left and right regions. We first establish that the ground state projector of
  any locally gapped frustration-free 1D spin system can be
  approximated to within error $\epsilon$ by a degree
  $O(\sqrt{n\log(\epsilon^{-1})})$ multivariate polynomial in the
  interaction terms of the Hamiltonian. This generalizes the optimal bound on the approximate degree of
  the boolean AND function, which corresponds to the special case of
  commuting Hamiltonian terms. For 2D spin systems we then
  construct an approximate ground state projector (AGSP) that
  employs the optimal 1D approximation in the vicinity of the
  boundary of the bipartition of interest. This AGSP has
  sufficiently low entanglement and error to establish
  the area law using a known technique.
\end{abstract}

\section{Introduction} 

The information-theoretic view on quantum matter has had widespread
impact in physics. For instance, tools from quantum Shannon theory
have provided insights into the black-hole paradox \cite{HaydenP07}
and the notion of topological entanglement entropy has been crucial
for understanding and classifying phases of matter \cite{KitaevP06}.
This viewpoint has also permeated the computational side of
condensed matter physics, and has led to the identification of
entropic properties known as \textit{the area laws}, which are
hallmarks of classical simulability in many physically relevant
settings. A state of a lattice spin system is said to satisfy an
area law if its entanglement entropy with respect to any bipartition
scales with the size of its boundary. This restricts the quantum
correlations arising in the state, and enables an efficient
classical representation of the state for one-dimensional (1D)
lattice systems \cite{Vidal03}. A breakthrough result of Hastings
\cite{Hastings07} established an area law for gapped ground states
in one dimension. Subsequent improvements were obtained in
Refs.~\cite{AradLV12, AradKLV13} using new tools from combinatorics
and approximation theory. This led to an efficient classical
algorithm for computing ground states \cite{LandauVV13, AradLVV17},
providing a rigorous justification for the success of the DMRG
algorithm \cite{White93}. Recently, area law was also shown for
the ground states of 1D long range hamiltonians \cite{KuwaharaS20}.

The area law conjecture for two or higher dimensional systems has
remained a significant open question, see e.g.,
Refs.~\cite{eisert2008area, PhysRevA.80.052104,
PhysRevLett.113.197204, PhysRevB.92.115134}. It can be motivated by
the following ``locality intuition": 

\vspace{0.05in}

\textit{Locality of correlations in the vicinity of the boundary of
a region \underline{should} imply an area law for the region}. 

\vspace{0.05in}

In particular, this suggests that the area law should hold for
gapped ground states since they possess a finite correlation length
\cite{Hastings04}. Whether or not this intuition can be made
rigorous remains to be seen \cite{Hastings07a}. While
correlation decay has been shown to imply an area law in 1D
\cite{BrandaoH13}, the locality intuition has no formal support in
higher dimensions.

In this work, we prove that the unique ground state of any
\textit{frustration-free, locally gapped} 2D lattice spin system
satisfies an area law scaling of entanglement entropy with respect
to a vertical cut that partitions the system into left and right
regions, see \Fig{fig:set-up}.
\begin{theorem}[Informal]
  Consider a locally gapped, frustration-free Hamiltonian with
  geometrically local interactions in 2D and a unique
  ground state. The ground state entanglement entropy with respect
  to a vertical bipartition of length $n$ is at most
  $n^{1+O\left(\frac{1}{\mathrm{poly}{\log(n)}}\right)}$. 
  \label{thm:informalarea}
\end{theorem}

\begin{figure}
\centering
\begin{tikzpicture}[xscale=0.45,yscale=0.45]

\draw [fill=blue!10!white, thick] (0.5,0.5) rectangle (10.5, 8.5);
\draw [fill=blue!30!white] (0.5, 0.5) rectangle (5.5,8.5);
\draw [blue, ultra thick] (5.5, 8.9)--(5.5,0.1);
\node at (5, 9){\footnotesize $c$};
\node at (6.5, 9){\footnotesize $c+1$};
\foreach \i in {1,...,10}
{
\foreach \j in {1,...,8}
   \draw (\i, \j) node[circle, fill=black, scale=0.25]{};
}

\draw [fill=white!80!green, thick] (1,1) rectangle (2,2);
\node at (1.5, 1.5) {\footnotesize $h_{ij}$};

\draw [<->] (0.5,0) -- (10.5, 0);
\node at (5.5,-0.5) {\footnotesize $L$};
\draw [<->] (0,0.5) -- (0, 8.5);
\node at (-1, 4.5) {\footnotesize $n+1 $};
\end{tikzpicture}
\caption{\small $L\times (n+1)$ lattice, with local terms $h_{ij}$
acting on plaquettes as depicted in green. A vertical
cut (blue line) partitions the system into left and right
regions, and intersects $n$ plaquettes. }
\label{fig:set-up}

\end{figure}
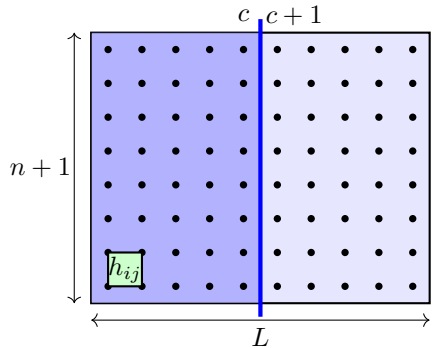
A frustration-free quantum spin system has the property that its
ground state has minimal energy for each term in the Hamiltonian.
Such a system is said to be locally gapped if there exists a
positive constant that lower bounds the spectral gap of any
subset of the local Hamiltonian terms.  
We believe that our techniques readily generalized to rectangular bi-partitions
on the lattice. This can then be used to prove area laws for other bi-partitions via appropriate tiling, including those
featuring gapless edge excitations \cite{BachmannHN15}. 
Using the techniques introduced in \cite{AradLVV17} and further
developed in \cite{Abrahamsen2020}, Theorem~\ref{thm:informalarea}
readily extends to degenerate ground states. We note that a previous
work~\cite{BeaudrapOE10} established the area law for the special
case of spin-$1/2$ frustration-free systems, using an exact
characterization of the ground space from
\cRef{bravyi2011efficient}. The area law is known to be false
on general graphs \cite{aharonov2014local, AHS20}.

The proof of \Thm{thm:informalarea} is obtained via new insights in
the approximation theory of quantum ground states. For a Hamiltonian
with unique ground state $|\Omega\rangle$, an $\epsilon$-approximate
ground state projector (AGSP) is an Hermitian operator $K$ such that
$K|\Omega\rangle=|\Omega\rangle$ and $\|K-|\Omega\rangle\langle
\Omega|\|\leq \epsilon$. In \cRef{AradLV12} it has been shown that
an AGSP with small error $\epsilon$ and low entanglement with
respect to a given bipartition of the lattice --- as measured by its
Schmidt rank $\mathrm{SR}(K)$\footnote{The Schmidt rank $\SR(K)$ of
an operator $K$ with respect to a bipartition of the system is the
minimal number $R$ of tensor product operators $A_\alpha\otimes
B_\alpha$ such that $K=\sum_{\alpha=1}^R A_\alpha\otimes B_\alpha$.}
--- implies an upper bound on the entanglement of the ground state
itself across the bipartition. In particular, \cRef{AradLV12} shows
that if $\epsilon\cdot \mathrm{SR}(K)\leq 1/2$ then the entanglement
of the ground state is at most $O(1)\cdot
\mathrm{log(\mathrm{SR}(K))}$, see \Thm{thm:agsp} for a precise
statement. In this way the study of entanglement in quantum ground
states can be reduced to the study of entanglement in low-error
AGSPs. 

\paragraph*{Approximate degree of quantum ground states}

To describe our techniques, consider a collection
$\{H_j\}_{j=1}^{n}$ of Hermitian operators satisfying $0\leq H_j\leq
I$ for all $j$.  Suppose the Hamiltonian $H\EqDef\sum_{j=1}^{n}H_j$
has a unique ground state $|\Omega\rangle$ satisfying
$H_j|\Omega\rangle=0$ for all $j$. 

We consider AGSPs which are multivariate polynomial functions 
\begin{align}
  K=P(H_1,H_2,\ldots, H_n).
\label{eq:multipol}
\end{align}
In general, this kind of polynomial is a linear combination of
monomials of the form
\begin{align*}
  H_{j_1}H_{j_2}\ldots H_{j_m} \qquad \quad j_k\in [n].
\end{align*}
As discussed above, in order to bound the entanglement of the ground
state, it suffices to construct an AGSP with sufficiently small
error and sufficiently small entanglement. Moreover, the techniques of Ref.~\cite{AradKLV13} \textit{suggest}
that polynomial degree can be taken as a proxy for entanglement in
AGSPs of this type. Thus, we ask: what is the minimal polynomial
degree $s$ needed to approximate the ground state projector to
within a given error $\epsilon$? 

Following \cRef{AradKLV13}, it is instructive to consider the
special case in which our AGSP \eqref{eq:multipol} can be
expressed as $K=p(H)$ where $p$ is a \textit{univariate} polynomial.
This kind of AGSP has the nice feature that it commutes with $H$ and
can therefore be diagonalized in the same basis. Using this fact we
see that such a polynomial is an $\epsilon$-AGSP iff
\begin{align}
  p(0)=1 \qquad \text{and} \qquad 
    \max_{x\in \mathrm{\mathrm{Spec}_{+}(H)}} |p(x)|\leq \epsilon.
\label{eq:univariate}
\end{align}
Here $\mathrm{Spec}_{+}(H) $ is the set of nonzero eigenvalues of
$H$. Note that since $0\leq H_j\leq I$ we have $\|H\|\leq n$ and
therefore $\mathrm{Spec}_{+}(H) \subseteq [\gamma,n]$ where $\gamma$
is the smallest nonzero eigenvalue or spectral gap of $H$. By
choosing $p$ to be a rescaled and shifted Chebyshev polynomial of
degree $s$ one obtains an AGSP with \cite{AradKLV13}
\begin{align}
  \epsilon=e^{-\Omega(s\sqrt{\frac{\gamma}{n}})}.
\label{eq:chebtradeoff}
\end{align}
This scaling of error with degree is optimal, a consequence of
the extremal property of Chebyshev polynomials \cite[Proposition
2.4]{SachdevaV14}. Here we did not use any properties of the
Hamiltonian except the fact that $\mathrm{Spec}(H) \subseteq
[\gamma,n]$. We see that a spectral gap lower bounded by a positive
constant ensures a good $\epsilon=O(1)$ approximation by a
$O(\sqrt{n})$-degree polynomial. This form of locality in the ground
state is somewhat distinct from finite correlation length. 

Remarkably, the tradeoff \Eq{eq:chebtradeoff} between polynomial
degree and error can be improved in certain important special cases.
For example, suppose the Hamiltonian terms are commuting projectors,
i.e., $[H_i,H_j]=0$ and $H_i^2=H_i$.  In that case the problem of
approximating the ground state is formally equivalent to the problem
of approximating the multivariate AND function of $n$ binary
variables (equivalently, the OR function), see \Sec{sec:commute} for
details\footnote{The $\epsilon$-approximate degree of AND has the remarkable low-error behaviour
$\widetilde{\mathrm{deg}}_{\epsilon}(\mathrm{AND}) =
O(\sqrt{n\log(\epsilon^{-1})})$ \cite{BuhrmanCWZ99, de2008note}.
The log under the square root reflects the fact that the error
probability of Grover's search algorithm can be reduced using a
better strategy than the naive parallel amplification.}.  The
distinguishing feature of the commuting case for our purposes is
that, crucially, all eigenvalues of $H$ are integers between $0$ and
$n$ and by exploiting the fact that
$\mathrm{Spec}_{+}(H)\subseteq\{1,2,\ldots, n\}$ one can construct a
suitable univariate polynomial $p$ that achieves \Eq{eq:univariate}
with 
\begin{align}
\epsilon=e^{-\Omega(s^2/n)}.
\label{eq:better}
\end{align}
This improves upon \Eq{eq:chebtradeoff} in the low-error regime
$s\gg\sqrt{n}$ and is known to be optimal in the commuting case
\cite{KahnLS96}.

More generally, for a collection of possibly non-commuting operators
$\{H_j\}_{j=1}^{n}$ let us call a degree-$s$ multivariate polynomial
AGSP \eqref{eq:multipol} \textit{optimal} if the approximation error
matches \Eq{eq:better}. Our first result establishes that
one-dimensional frustration-free locally-gapped ground
states can be optimally approximated in this sense. 
\begin{theorem}[Optimal approximation of 1D ground states, informal]
  For any constant $\delta\in (0,1/2)$ and $s\in (\sqrt{n}, 
  n^{1-\delta})$, there is a degree $O(s)$ polynomial which 
  approximates the ground state projector of a locallycally gapped
  1D frustration-free quantum spin system to within error
  \Eq{eq:better}. \label{thm:informal1d}
\end{theorem}
We emphasize that the AGSP in the above theorem is a multivariate
polynomial of the form \eqref{eq:multipol}, and as far as we know
cannot be expressed as a univariate polynomial function of $H$. This
is because we may have $[H_i,H_{i+1}]\neq 0$, and --- in contrast
with the commuting case --- the spectrum $\mathrm{Spec}_{+}(H)$ does
not appear to have a nice characterization that allows us to improve
upon \Eq{eq:chebtradeoff} by a suitable choice of univariate
polynomial $p$. We construct our AGSP via a recursive application of
the robust polynomial method of \cRef{AAG19} (where a subvolume law
for the same class of systems was shown). The resulting polynomial,
which is detailed in \Sec{sec:mergprop}, has a structure which is
reminiscent of a renormalization group flow.

Although it concerns 1D systems, \Thm{thm:informal1d} turns out to
be just what we need to establish the area law in two dimensions.
The key insight is captured by the following modified locality
intuition that we propose, which asserts a direct link between
linear-degree optimal polynomial approximations and area laws:

\vspace{0.05in}

\textit{A linear-degree optimal polynomial
approximation for the ground state in the vicinity of the boundary
of a region \underline{should} imply an area law for the region}. 

\vspace{0.05in}

Here we mean linear in $n$, the number of inputs of the multivariate
polynomial (cf. Eq.~\eqref{eq:multipol}). To understand where this
comes from, suppose we can construct an optimal linear-degree
polynomial $P$ that approximates the ground state projector and is
localized in a width $\sim w$ neighborhood of the boundary of the
bipartition of interest (here we are intentionally vague about the
meaning of localized, see \Sec{sec:2d} for details). Thus, the
degree of $P$ is $\sim w\cdot \mathrm{area}$ and its error is
$\epsilon\leq e^{-\Omega(w\cdot \mathrm{area})}$, where
`$\mathrm{area}$' is the size of the boundary. Now consider an AGSP
$K=P^q$ for some positive integer $q$. The total polynomial degree
of $K$ is thus $D=qw\cdot \mathrm{area}$, and its error is
$\epsilon'=\epsilon^q\leq e^{-\Omega(qw\cdot \mathrm{area})}$. Now
we shall assume that the polynomial $K^q$ is nicely behaved in a
certain sense first identified in \cRef{AradLV12}. In particular, we
assume that its Schmidt rank is \textit{amortized} over the width
$w$ neighborhood of the boundary, in that it scales as
$\mathrm{SR}\sim e^{O(D/w+w\cdot \mathrm{area})}$. Choosing
$q=\Omega(w)$ (say) and letting $w$ be a large constant, we can
ensure $\epsilon'\cdot \mathrm{SR}\leq 1/2$, with
$\log(\mathrm{SR})=O(\mathrm{area})$. Thus, applying the
aforementioned method from \cRef{AradLV12} we would obtain the
desired upper bound $O(\mathrm{area})$ on the ground state
entanglement entropy.

Since the boundary of a region on a 2D lattice is one-dimensional,
the above argument suggests that the 2D area law should follow from
optimal linear-degree polynomial approximations in 1D. To make this
work, in \Sec{sec:2d} we show how our 1D approximation can be used
``in the vicinity of the boundary of the region" and we relate the
entanglement of the resulting AGSP to the polynomial degree. The
area law is then established using the aforementioned method from
\cRef{AradLV12}. The astute reader may note that
\Thm{thm:informal1d} does not quite provide a linear-degree optimal
polynomial as the degree must be $n^{1-\delta}$ for some $\delta\in
(0,1/2)$; a careful treatment of the $\delta\rightarrow 0$ limit
leads to the slight deviation $n^{1+o(1)}$ from area law behaviour
in \Thm{thm:informalarea}, see \Sec{sec:AGSP} for details.

\paragraph*{Discussion} 
There are at least three significant questions left open by our
work. Firstly, one may ask if the assumption of a local spectral gap
can be removed or replaced with one concerning the global spectral
gap of the Hamiltonian. We believe that this could lead to a
generalization of our techniques to frustrated systems. To make
progress here may require a deeper understanding of the interplay
between the local spectral gap and the gap of the full hamiltonian.
Secondly, it is natural to ask if ground states of locally gapped
frustration-free systems can be approximated by efficiently
representable tensor networks such as PEPS of small bond dimension
\cite{verstraete2008matrix}. While it is known that a 2D area law
does not imply such a representation \cite{GeE16}, a more detailed
study of the optimal polynomial approximations considered here may
provide insight into this question. Finally, a natural open question
is to extend our results to local hamiltonian systems on higher
dimensional lattices. As mentioned earlier, this is closely related
to the question of approximating ground states by linear-degree
optimal polynomials.

The rest of the paper is organized as follows. In \Sec{sec:toolkit}
we review some basic approximation tools, the Chebyshev and robust
polynomials. In \Sec{sec:AGSP} we deploy them to construct optimal
polynomial AGSPs for a family of quantum systems that includes 1D
quantum spin systems and we establish Theorem \ref{thm:informal1d}.
Finally, in \Sec{sec:2d} we adapt our methods to the 2D setting and
prove the area law Theorem \ref{thm:informalarea}.

\section{Polynomial approximation toolkit}
\label{sec:toolkit}

In this section we describe methods for approximating multivariate
functions by polynomials. We first describe polynomial
approximations with real-valued variable inputs. Then we generalize
these methods to the local Hamiltonian setting by allowing
operator-valued inputs.

\subsection{Approximation of functions}
\label{sec:approxfunc}

Following \cRef{AAG19} we shall build polynomial approximations by
combining two well-known ingredients: the univariate Chebyshev
polynomials and a robust polynomial \cite{Sherstov12}.  

We will use a rescaled and shifted Chebyshev polynomial defined as
follows. For every $s\in \mathbb{R}_{\geq 0}$ and $\eta\in (0,1)$ we
define a polynomial $T_{\eta,s}:[0,1]\rightarrow \mathbb{R}$ of
degree $\lceil s\rceil$ by
\begin{align}
  T_{\eta,s}(x) \EqDef \frac{T_{\lceil s\rceil}
    \left(\frac{2(1-x)}{1-\eta}-1\right)}{T_{\lceil s\rceil}
      \left(\frac{2}{1-\eta}-1\right)},
\label{eq:chebydef}
\end{align}
where $T_{j}$ is the Chebyshev polynomial of the first kind. To ease
notation later on, the parameter $s$ which determines the degree is
not required to be an integer.  The polynomial \Eq{eq:chebydef} has
the following property which is a direct consequence of Lemma 4.1 of
Ref.~\cite{AradKLV13}.
\begin{lemma}[\cite{AradKLV13}]
  For every $s\in \mathbb{R}_{\geq 0}$ and $\eta\in (0,1)$ we have
  $T_{\eta,s}(0)=1$ and
  \begin{align*}
    |T_{\eta,s}(x)|\leq 2 e^{-2s\sqrt{\eta}} \qquad \eta\leq x\leq 1.
  \end{align*}
\label{lem:cheby}
\end{lemma}
Next, we describe the robust polynomial. Our starting point is 
the function $B: [0,1]\rightarrow\{0,1\}$ defined by
\begin{align}
  B(x)\EqDef \begin{cases} 1, & x=1\\ 0, & 0\leq x<1.\end{cases}.
\label{eq:bit}
\end{align}
This function rounds $x$ to a bit in a one-sided fashion. Using
\eqref{eq:bit} we define a (one-sided) ``robust product" that takes
real inputs $x_1,x_2,\ldots, x_m\in [0,1]$ and outputs $1$ if and
only if they are all equal to 1:
\begin{align}
  \Rob(x_1,x_2,\ldots, x_m)\EqDef B(x_1)B(x_2)\ldots B(x_m).
\label{eq:and}
\end{align}
We note that since $B(x_j)^2=B(x_j)$ we may also express \Eq{eq:and} as 
\begin{align}
\Rob(x_1,x_2,\ldots, x_m) = \big(B(x_m)B(x_{m-1})\ldots B(x_1)\big)
  \big( B(x_1)B(x_2)\ldots B(x_m)\big).
\label{eq:and2}
\end{align}
The left-right symmetric expression will be useful to us momentarily
when we extend the definition of the function to allow
operator-valued inputs.

The robust polynomial of interest is an approximation to
\Eq{eq:and}. To this end, let
\begin{align}
  B_i(x)\EqDef \begin{cases} 
                  x, & i=1\\ 
                  x^{i-1}(x-1), &  2\leq i\leq \infty.
               \end{cases}
\label{eq:bi}
\end{align}
Note that for any $x\in [0,1]$ we have
$B(x)=\mathrm{lim}_{i\rightarrow \infty} \sum_{j=1}^{i} B_j(x)$.
Define
\begin{align}
  \widetilde{\Rob}(x_1,x_2,\ldots, x_m)
    \EqDef \sum_{(i_1+i'_1)+\ldots+(i_m+i'_m)\leq 3m}
      \!\!\!\!\!\!\!\!\!\!\!\!\!\!\!\!
      \left(B_{i'_m}(x_m)\ldots, B_{i'_1}(x_1) \right)
      \left(B_{i_1}(x_1)\ldots, B_{i_m}(x_m)\right).
\label{eq:robdefclass}
\end{align}
The above expression is obtained by starting with \Eq{eq:and2},
substituting $B\leftarrow \sum_{j=1}^{\infty} B_j$, and then
truncating the summation so that the total degree of the polynomial
is at most $3m$ (this is somewhat arbitrary). This
polynomial is a good approximation to $\Rob$ in the
following sense.
\begin{lemma}[Special case of Lemma \ref{lem:roblem}]
\label{lem:roblemclass} 
  Suppose that $x_1,x_2,\ldots, x_m\in [0,\eps]\cup \{1\}$ for some
  $\epsilon\leq 1/10$. Then $|\widetilde{\Rob}(x_1,\ldots, x_m) - \Rob(x_1,\ldots ,x_m)|\leq \br{10 \eps}^{m}.$
\end{lemma}

\subsection{Approximation of operators}

Let us now extend our definitions from the previous section to allow
operator-valued inputs. Suppose $O$ is a Hermitian operator with all
eigenvalues in the interval $[0,1]$. The operator-valued Chebyshev
polynomial $T_{\eta,s}(O)$ is defined in the usual way by
substituting $x\leftarrow O$ in \Eq{eq:chebydef}.  By applying Lemma
\ref{lem:cheby} to each eigenvalue of $O$, we obtain the following.

\begin{lemma}
  Let $s\in \mathbb{R}_{\geq 0}$ and $\eta\in (0,1)$. Suppose that
  $O$ is an Hermitian operator with eigenvalues in the interval
  $\{0\}\cup [\eta, 1]$ and let $\Pi$ be the projector onto the
  nullspace of $O$. Then $T_{\eta,s}(O)\Pi=\Pi$ and $ \|T_{\eta,s}(O)-\Pi\|\leq 2 e^{-2s\sqrt{\eta}}.$ 
  \label{lem:chebyop}
\end{lemma}

For the robust polynomial, we start by defining $B(O)$ to be the
projector onto the eigenspace of $O$ with eigenvalue $1$.  For
Hermitian operators $O_1,O_2,\ldots, O_m$ such that each $O_i$ has
eigenvalues in the interval $[0,1]$, we define a Hermitian robust
product which generalizes \Eq{eq:and2}:
\begin{align*}
  \Rob(O_1,O_2,\ldots, O_m)
    \EqDef C^{\dagger}C \qquad \text{where} 
      \qquad C\EqDef B(O_1)B(O_2)\ldots B(O_m).
\end{align*}
Note that due to the possible non-commutativity of the $\{O_i\}$
operators, $\Rob(O_1,O_2,\ldots, O_m)$ is generally \emph{not} the
projector onto the intersection of the $+1$ eigenspaces of these operators.

We also define the operator-valued polynomial $B_i(O)$ for positive
integers $i$ by substituting $x\leftarrow O$ in \Eq{eq:bi}. The
robust polynomial is defined in parallel with
\eqref{eq:robdefclass}, i.e., 
\begin{align}
  \widetilde{\Rob}(O_1,O_2,\ldots, O_m)
    \EqDef \!\!\!\!\!\!\!\!\!\!\!\!\!\!\!\!
    \sum_{(i_1+i'_1)+\ldots+(i_m+i'_m)\leq 3m}
    \!\!\!\!\!\!\!\!\!\!\!\!\!\!\!\!
      \left(B_{i'_m}(O_m)\ldots, B_{i'_1}(O_1) \right)
      \left(B_{i_1}(O_1)\ldots, B_{i_m}(O_m)\right).
\label{eq:robdef}
\end{align}
One can easily check that the operator in \Eq{eq:robdef} is Hermitian.
Let us now establish the following error bound.
\begin{lemma}
\label{lem:roblem} 
  Suppose that the eigenvalues of all operators $\{O_i\}_{i=1}^m$
  lie in the range $[0,\eps]\cup \{1\}$ for some $\epsilon\leq
  1/10$. Then
  \begin{align*}
    \|\widetilde{\Rob}(O_1,\ldots ,O_m)-\Rob(O_1,\ldots ,O_m)\|
      \leq \br{10 \eps}^{m}.
  \end{align*}
\end{lemma}

\begin{proof}
  To ease notation in this proof, we use the shorthand
  $\vec{i}\EqDef (i_1,i'_1,i_2,i'_2,\ldots, i_m,i'_m)$ for a tuple
  of $2m$ positive integers, $\mathrm{sum}(\vec{i})\EqDef
  i_1+i'_1+i_2+i'_2+\ldots + i_m+i'_m$ for their sum , and 
  \begin{align*}
    M(\vec{i})\EqDef B_{i'_m}(O_m)\ldots B_{i'_1}(O_1)B_{i_1}(O_1)
      \ldots B_{i_m}(O_m)
  \end{align*}
  for the product that appears in \Eq{eq:robdef}. Using \Eq{eq:bi}
  we see that for any Hermitian operator $O$ with eigenvalues in
  $[0,\eps]\cup \{1\}$ we have $\|B_i(O)\|\leq \epsilon^{i-1}$ for
  all $i\geq 1$. Therefore $\|M(\vec{i})\|\leq
  \epsilon^{\mathrm{sum}(\vec{i})-2m}$ and, for any integer $q\geq
  2m$, 
  \begin{align}
\bigg\|\sum_{\vec{i}:\mathrm{sum}(\vec{i})=q} 
      M(\vec{i})\bigg\|\leq \sum_{\vec{i}:\mathrm{sum}(\vec{i})=q} 
        \epsilon^{q-2m}={{q-1}\choose{2m-1}}  
        \epsilon^{q-2m}\leq 2^q\epsilon^{q-2m}.
  \label{eq:qtail}
  \end{align}
  Here we used the fact that each component of $\vec{i}$ is a
  \textit{positive} integer and so $\vec{i}$ is a composition of $q$
  with exactly $2m$ parts; the number of compositions of an integer
  $n$ with $k$ parts is ${{n-1}\choose{k-1}}$. In the last
  inequality in \Eq{eq:qtail} we used the upper bound ${{a}\choose
  {b}}\leq 2^a$. Now let $J>3m$ be a positive integer to be fixed
  later. We have
  \begin{equation*}
    \sum_{ i_1,i'_1,\ldots, i_m, i'_m =1}^J\!\!\!\!\!\!\!\!
      M(\vec{i})-\widetilde{\Rob}(O_1,\ldots, O_m)
      =  \sum_{q=3m+1}^{2mJ} \sum_{\substack{\vec{i}:\mathrm{sum}(\vec{i})=q 
            \\ i_1,i_1',\ldots, i_m, i_m' \leq J} } M(\vec{i})
  \end{equation*}
  and therefore
  \begin{equation}
    \bigg\| \sum_{ i_1,i'_1,\ldots, i_m, i'_m =1}^J\!\!\!\!\!\!\!\!
      M(\vec{i})-\widetilde{\Rob}(O_1,O_2,\ldots, O_m)\bigg\|
      \leq \epsilon^{-2m}\sum_{q=3m+1}^{2mJ} 
          \left(2\epsilon\right)^{q}
      \leq (8\epsilon)^m \big(\frac{2\epsilon}{1-2\epsilon}\big)
        \leq \left(8\epsilon\right)^m,
  \label{eq:m1}
  \end{equation}
  where in the last step we used $\epsilon\leq1/10$. Using
  \Eq{eq:bi} gives $\|\sum_{1\leq i\leq J} B_i(O)-B(O)\|\leq
  \epsilon^J$ for any Hermitian $O$ with eigenvalues in
  $[0,\eps]\cup \{1\}$. Applying this bound $2m$ times and using the
  triangle inequality gives
  \begin{align}
    \bigg\| \sum_{ i_1,i'_1,\ldots, i_m, i'_m =1}^J M(\vec{i})
      - \Rob(O_1,\ldots, O_m)\bigg\|\leq 2m\epsilon^J.
  \label{eq:m2}
  \end{align}
  Let us choose $J$ to be large enough so that the right-hand-side
  of \Eq{eq:m2} is at most $(10\epsilon)^m-(8\epsilon)^m$. Combining
  Eqs.~(\ref{eq:m1}, \ref{eq:m2}) using the triangle inequality
  completes the proof. 
\end{proof}

We also use the following claims which summarize simple properties
of $\widetilde{\Rob}$.
\begin{claim}
  Let $O_1, O_2\ldots ,O_m$ be Hermitian, with eigenvalues in the
  range $[0,1]$. Suppose there exists a projector $\Pi$ such that
  $O_j\Pi=\Pi$ for all $j\in [m]$. Then $\widetilde{\Rob}(O_1,
  \ldots O_m)\Pi=\Pi$. \label{clm:robground} 
\end{claim} 

\begin{proof}
  Note that for all $j\in [m]$, $B_1(O_j)\Pi=O_j\Pi=\Pi$ and for all
  $i\geq 2$, $B_i(O_j)\Pi=0$. Thus,
  \begin{align*}
    \widetilde{\Rob}(O_1, \ldots O_m)\Pi
      = B_1(O_m)\ldots B_1(O_1)B_1(O_1)\ldots B_1(O_m)\Pi
      = \Pi.
  \end{align*}
\end{proof}

\begin{claim}
\label{claim:unwrap} 
  The polynomial $\widetilde{\Rob}(O_1, \ldots ,O_m)$ can be
  expressed as a linear combination of at most $2^{5m}$ monomials of
  the form 
  \begin{align}
    O_{i_1}^{a_1}O_{i_2}^{a_2}\ldots O_{i_{2m}}^{a_{2m}}.
  \label{eq:monom}
  \end{align}
  where $i_1,i_2,\ldots, i_{2m}\in [m]$ and $\{a_i\}$ are positive
  integers satisfying $\sum_{j=1}^{2m} a_{j}\leq 3m$.
\end{claim}

\begin{proof}
  The definition \Eq{eq:robdef} expresses $\widetilde{\Rob}(O_1,
  \ldots ,O_m)$ as a sum of ${ {3m-1} \choose {2m-1}}\leq
  2^{3m-1}\leq 8^m$ operators 
  \begin{align}
    B_{i'_m}(O_m)\ldots B_{i'_1}(O_1)B_{i_1}(O_1)\ldots B_{i_m}(O_m) 
      \qquad (i_1+i'_1)+\ldots (i_m+i'_m)\leq 3m.
  \label{eq:bterms}
  \end{align}
  Now observe from \Eq{eq:bi} that $B_i(O)\in \{O, O^{i}-O^{i-1}\}$.
  Therefore each term \Eq{eq:bterms} can be expanded as a sum of at
  most $2^{2m}$ monomials of the form \Eq{eq:monom}. Therefore
  $\widetilde{\Rob}(O_1, \ldots O_m)$ expands into at most $8^m\cdot
  2^{2m}= 2^{5m}$ terms of the form \Eq{eq:monom}.
\end{proof}

\section{Optimal ground state approximations}
\label{sec:AGSP}

Throughout this section we consider the following scenario. We are
given a set of Hermitian operators $\{H_j\}_{j=1}^{n}$ such that
\begin{align}
  0\leq H_j \leq I \qquad \text{for all $j\in [n]$},
\label{eq:noncom}
\end{align}
which act on some finite-dimensional Hilbert space
$\mathcal{H}$. We are interested in the nullspace of the operator 
\begin{align*}
H=\sum_{j=1}^{n} H_j.
\end{align*}
Let us write $\Pi$ for the projector onto the nullspace of $H$. In
other words, $\Pi$ projects onto the intersection of the nullspaces
of all operators $H_j$ (we are interested in the case where $\Pi$ is
nonzero). Our goal is to approximate $\Pi$ by a low-degree
polynomial in the operators $\{H_j\}$.

In \Sec{sec:commute} and \Sec{sec:mergprop} we work in a general
setting and in particular we do not assume a tensor product
structure of the Hilbert space $\mathcal{H}$ or geometric locality
of the operators $\{H_j\}$. In \Sec{sec:commute} we consider the
simplest case in which all operators $H_j$ are mutually commuting
and we describe the known optimal tradeoff between approximation
degree and error. Then, in \Sec{sec:mergprop} we show that optimal
approximations can be obtained more generally for noncommuting
operators which satisfy certain \textit{gap} and \textit{merge}
properties. These properties themselves assert a kind of
one-dimensional structure with respect to the given ordering $1\leq
j\leq n$ of the operators. In \Sec{sec:1D} we describe how a direct
application of these results provides optimal ground state
approximations for one-dimensional, locally-gapped, frustration-free
qudit Hamiltonians. Later we will see how the results of 
\Sec{sec:mergprop} can provide low-entanglement approximations of
ground states in the 2D setup.

\subsection{Commuting projectors}
\label{sec:commute}

We begin with the easy case in which all $\{H_i\}$ are commuting
projectors:
\begin{align}
  H_i^2=H_i \qquad \text{and} \qquad [H_i, H_j]=0 
    \qquad \text{for all} \qquad i,j\in [n].
\label{eq:comm}
\end{align}
In this case $(I-H_i)$ is the projector onto the nullspace of $H_i$,
and due to the commutativity \Eq{eq:comm} we may express $\Pi$
exactly as the degree-$n$ polynomial $\Pi=\prod_{i=1}^{n} \left(I-H_i\right)$. Our goal is to construct a lower degree polynomial $P$ that
approximates $\Pi$. Since all operators $\{H_i\}$ commute and have
$\{0,1\}$ eigenvalues, we may work in a basis in which they are
simultaneously diagonal and the problem reduces to that of
approximating the product of binary variables $x\in \{0,1\}^n$ which
label the eigenvalues of $\{I-H_i\}_{i=1}^{n}$. (Note that here we
do not require any properties of the basis which simultaneously
diagonalizes these operators, only that it exists). In other words, the problem of approximating the ground space
projector for a Hamiltonian which is a sum of commuting projectors,
reduces to the problem of approximating the boolean AND function
\begin{align*}
  \mathrm{AND}(x_1,x_2,\ldots, x_n)
    = \begin{cases}1 & \text{if }x_1=x_2=\ldots =x_n=1\\ 
    0 & \text{otherwise}\end{cases}.
\end{align*}
We are faced with the task of constructing a multilinear polynomial
$p$ which $\epsilon$-approximates $\mathrm{AND}$ in the sense that
$|p(x)-\mathrm{AND}(x)|\leq \epsilon$ for each $x\in \{0,1\}^n$.
Remarkably, it is possible to achieve an arbitrarily small constant
error $\epsilon=O(1)$ using a polynomial of degree $O(\sqrt{n})$
\cite{linial1990approximate}. For example, one can use the Chebyshev polynomial $T_{1/n,
s}\big(\frac{1}{n}\sum_{i=1}^{n} x_i\big)$ of degree $\lceil
s\rceil$ which achieves an approximation error $\epsilon=
e^{-\Omega\left( s/\sqrt{n}\right)}$ as can be seen from Lemma
\ref{lem:cheby}. Similarly, the acceptance probability of the
standard Grover search algorithm \cite{grover1996fast}, viewed as a
function of the input bit string $x$ provided as an oracle,
constructs such an approximating polynomial \cite{beals2001quantum}.
However, neither of these polynomials has optimal degree in the
low-error regime where $\epsilon$ decreases with $n$. In that regime
an optimal polynomial can be constructed via a low-error refinement
of Grover search \cite{BuhrmanCWZ99, de2008note} (see also
\cRef{KahnLS96}).

Here we provide a different family of polynomials that give an
optimal approximation to the $\mathrm{AND}$ function. These
polynomials are obtained in a simple way by combining the Chebyshev
polynomial $T_{\eta,s}$ and the robust polynomial $\widetilde{\Rob}$
from \Sec{sec:approxfunc}.  Soon we will see how this
construction can be extended to the non-commuting case. It is
unclear to us whether one can alternatively extend the known optimal
polynomials constructed in Refs.~\cite{KahnLS96, BuhrmanCWZ99,
de2008note}.

\begin{theorem}[\textbf{Optimal approximation of AND}]
\label{thm:ANDapprox}
  Let $n$ be a positive integer. For every real number $s \in
  \br{\sqrt{n}, n}$, there exists a polynomial $P(x)$ of degree $O(s)$
  such that 
  \begin{align*}
    |P(x)-\mathrm{AND}(x)|=e^{-\Omega(\frac{s^2}{n})} 
      \quad \text{for all} \quad x\in \{0,1\}^n.
  \end{align*}
\end{theorem}

\begin{proof}
  Define the positive integer $t=\lceil \frac{n^2}{s^2}\rceil$ and
  note that $1\leq t\leq n$ due to the specified bounds on $s$. Let
  $p(y)\EqDef T_{\frac{1}{t},2\sqrt{t}}\br{y}$. From Lemma
  \ref{lem:cheby} we see that 
  \begin{align}
    p(0)=1 \qquad \text{and} \qquad  \left|p(y)\right|
      \leq 2\cdot e^{-4}\leq \frac{1}{20}
        \qquad \text{ for all } \quad \frac{1}{t}\leq y\leq 1.
  \label{eq:pprop}
  \end{align}

  Since $t\leq n$, we may construct a partition $[n]=I_1\cup
  I_2\cup\ldots \cup I_{\prt}$ where $\prt\EqDef \lceil n/t\rceil$
  and $|I_k|\leq t$ for all $1\leq k\leq \prt$. Our polynomial
  approximation to $\mathrm{AND}$ is defined as
  \begin{align*}
73A72EA3    P(x)=\widetilde{\Rob}\Bigg(p\bigg(1-\frac{1}{|I_1|}
      \sum_{j\in I_1}x_j\bigg), p\bigg(1-\frac{1}{|I_2|}
        \sum_{j\in I_2}x_j\bigg),\ldots, p\bigg(1-\frac{1}{|I_\prt|}
          \sum_{j\in I_{\prt}}x_j\bigg)\Bigg).
  \end{align*}
  Now we observe that the $k$th input to the $\widetilde{\Rob}$
  function on the RHS approximates the AND of all bits in the set
  $I_k$. To see this, note
  that $1-\frac{1}{|I_k|} \sum_{j\in I_k}x_j=0$ when $x_j=1$ for all
  $j\in I_k$, and $1-\frac{1}{|I_k|} \sum_{j\in I_k}x_j\ge 1/t$ if
  one or more $x_j=0$.  Using this fact and \Eq{eq:pprop}, we see that for each $1\leq
  k\leq \prt$ we have
  \begin{align}
    \bigg|p\bigg(1-\frac{1}{|I_k|}\sum_{j\in I_k}x_j\bigg)
      - \prod_{j\in I_k} x_j\bigg|\leq \frac{1}{20}.
  \end{align}
  Now applying Lemma \ref{lem:roblemclass} with $\eps= 1/20$, and
  noting that $\Rob(x)=\mathrm{AND}(x)$ we see that, for each $x\in
  \{0,1\}^n$,
  \begin{align*}
    |P(x)-\mathrm{AND}(x)|\leq 2^{-\prt}\leq 2^{-n/t}\leq 2^{-s^2/n}.
  \end{align*}
  The degree of the polynomial is at most $3\prt\cdot
  2\sqrt{t}=O(s)$.
\end{proof}

\subsection{Operators with gap and merge properties}
\label{sec:mergprop}

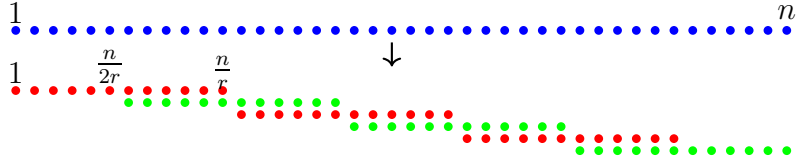
\begin{figure}
\centering
\begin{tikzpicture}[xscale=1,yscale=0.8]

\foreach \i in {1,...,42}
{
   \draw (\i/4, 0) node[circle, fill=blue, scale=0.3]{};
}
\draw [thick,->] (5.25,-0.2) -- (5.25,-0.6);

\node at (1/4,0.3) {$1$};

\node at (42/4,0.3) {$n$};

\foreach \i in {1,...,12}
{
   \draw (\i/4, -1) node[circle, fill=red, scale=0.3]{};
}

\node at (1/4,0.3-1) {$1$};

\node at (6/4,0.4-1) {$\frac{n}{2r}$};

\node at (12/4,0.3-1) {$\frac{n}{r}$};

\foreach \i in {7,...,18}
{
   \draw (\i/4, -1.2) node[circle, fill=green, scale=0.3]{};
}
 
\foreach \i in {13,...,24}
{
   \draw (\i/4, -1.4) node[circle, fill=red, scale=0.3]{};
}

\foreach \i in {19,...,30}
{
   \draw (\i/4, -1.6) node[circle, fill=green, scale=0.3]{};
}

\foreach \i in {25,...,36}
{
   \draw (\i/4, -1.8) node[circle, fill=red, scale=0.3]{};
}

\foreach \i in {31,...,42}
{
   \draw (\i/4, -2) node[circle, fill=green, scale=0.3]{};
}
\end{tikzpicture}
  \caption{\small An interval of length $n$ is decomposed into
  smaller intervals of length $n/r$ each. The overlap between
  consecutive intervals is exactly $\frac{n}{2r}$. The number of
  intervals is $\prt=2r-1$.} 
\label{fig:recursechain}
\end{figure}

We now consider a more general case in which the operators
$\{H_j\}_{j=1}^{n}$ still satisfy \eqref{eq:noncom}, but may not be
projectors and are not assumed to commute. For any subset
$S\subseteq [n]$ of the operators, we define the corresponding
Hamiltonian 
\begin{align*}
  H_S\EqDef \sum_{j\in S} H_j
\end{align*}
and the projector $\Pi_S$ onto its nullspace. Similarly, we define
$\mathrm{gap}(H_S)$ to be the smallest nonzero eigenvalue of $H_S$
\footnote{We use the convention that $\mathrm{gap}(h)=1$ if $h=0$.}.
A crucial difference between our setting here and the commuting
setting considered previously, is that a product $\Pi_S \Pi_T$ is
not in general equal to $\Pi_{S\cup T}$.

We require our operators to satisfy two properties which are
defined with respect to the given ordering $1\leq j\leq n$. To
describe these properties it will be convenient to define an
$\textit{interval}$ as a contiguous subset $\{j,j+1,\ldots, k-1,
k\}\subseteq [n]$. The \textit{gap property} states a lower bound
$\Delta$ on the smallest nonzero eigenvalue of any interval
Hamiltonian $H_S$. The \textit{merge property} asserts that $\Pi_S
\Pi_T\approxeq \Pi_{S\cup T}$ for overlapping intervals $S,T$, with
 error decreasing exponentially in the size of the overlap
region.  We now state these properties more precisely.
\begin{definition}
  Operators $\{H_j\}_{j=1}^{n}$ satisfy the gap and merge properties
  if, for some $\Delta\in (0,1]$, the following conditions hold for
  all intervals $S\subseteq [n]$ and any partition $S=ABC$ into three consecutive intervals:
  \begin{align}
    \mathrm{gap}(H_S)&\geq \Delta \qquad &\textbf{[Gap property]}
    \label{eq:locgap}\\
    \|\Pi_{AB}\Pi_{BC}-\Pi_{S}\|
      &\leq  2 e^{-|B|\sqrt{\Delta}} &\textbf{[Merge property]}.
  \label{eq:merge}
  \end{align}
\end{definition}
Note that the parameter $\Delta$ in this definition appears in both
the gap and merge properties. One could alternatively consider a
more general definition where each of these properties has its own
parameter, though we will not need to.

In the following we show that the optimal scaling
$e^{-\Omega\br{\frac{s^2}{n}}}$ of error with degree $s$ can be
recovered in this noncommutative setting, by a recursive use of the
robust polynomial, with one use of the Chebyshev polynomial and gap
property in the base level of the recursion. The analysis uses the
merge property to bound the error in the recursion.  The following
theorem describes our results for the case where the approximation
degree scales less than linearly in $n$.
\begin{theorem}[\textbf{Less than linear degree}]
\label{thm:1Dpoly} 
  Suppose $\{H_j\}_{j=1}^{n}$ satisfy Eqs.~(\ref{eq:noncom},
  \ref{eq:locgap},\ref{eq:merge}) for some $\Delta\in (0,1]$. Let
  $\delta\in (0,1/4)$ be fixed and let $s$ be a real number
  satisfying 
  \begin{align}
    2\sqrt{n}\Delta^{-1/2}\leq s \leq (1/4)n^{1-\delta}\Delta^{-1/4}.
  \label{eq:dbounds}
  \end{align}
  There is a degree $O(s)$ Hermitian multivariate polynomial $P$ in
  the operators $\{H_j\}_{j=1}^{n}$ such that
  \begin{align}
    P\Pi=\Pi \qquad \text{and} \qquad \|P-\Pi\| 
      = e^{-\frac{s^2 \Delta}{4n}}.
  \label{eq:papprox}
  \end{align}
\end{theorem} 
In the above, the big-O notation hides a constant which depends only
on $\delta$. We shall also be interested in a case where $\delta$ is
taken very close to $0$ and the degree is close to linear. This
almost-linear degree approximation will be used to establish the
area law for two-dimensional spin systems. For that application it
will be useful to describe the structure of the polynomial $P$ in
more detail. To this end, we first define certain families
$P(\alpha,\beta)$ of elementary polynomials as follows.
\begin{definition}
\label{def:pab} 
  For $\alpha,\beta >0$, let $\mathcal{P}(\alpha,\beta)$ denote the
  set of polynomials of the form
  \begin{align*}
    (H_{S_1})^{j_1} (H_{S_2})^{j_2}\ldots (H_{S_k})^{j_k} 
      \qquad \quad j_1+j_2+\ldots +j_k\leq \alpha 
      \quad \text{and} \quad k\leq \beta.
  \end{align*}
  where $j_1,j_2,\ldots, j_k$ are positive integers and each set
  $S_1,S_2,\ldots, S_k\subseteq [n]$ is an interval.
\end{definition}

\begin{theorem}[\textbf{Near-linear degree}]
\label{thm:almostlin}
  Suppose $\{H_j\}_{j=1}^{n}$ satisfy
  Eqs.~(\ref{eq:noncom},\ref{eq:locgap},\ref{eq:merge}) for some
  $\Delta\in (0,1]$ and that $n\geq C\Delta^{-1}$, where $C>0$ is
  some absolute constant. There exist real numbers
  \begin{align}
    \alpha\leq n\Delta^{-1/4} \qquad \text{and} \qquad 
      \beta=\Delta^{1/2}n^{1-O\left((\log n)^{-1/4}\right)}
  \label{eq:albet}
  \end{align}
  such that the following holds. There exists a Hermitian
  multivariate polynomial $P$ in the operators $\{H_i\}_{i=1}^{n}$
  of degree at most $\alpha$ such that
  \begin{align}
    P\Pi=\Pi \qquad \text{and} \qquad \|P-\Pi\| 
      \leq \exp{\left(-\beta e^{\sqrt{\log(n)}}\right)},
  \label{eq:pconds}
  \end{align}
  and such that $P$ can be expressed as a linear combination of at
  most $(2\alpha)^\beta$ elements of $\mathcal{P}(\alpha,\beta)$.
\end{theorem} 

Theorems~\ref{thm:1Dpoly} and \ref{thm:almostlin} are obtained as
consequences of the following Lemma, which treats the special case
where $n$ has the form $tr^{b-1}$ for suitably chosen positive
integers $t, r,b$. It constructs an approximating polynomial $P$
recursively, with $b$ levels of recursion. 
\begin{lemma}
  Suppose $n=tr^{b-1}$ for positive integers $t,r,b$ such that $t$ is even and
  \begin{align}
  \frac{16r}{t\sqrt{\Delta}}\leq 1/\Gamma\leq 1
  \label{eq:n1r}
\end{align}
for some real positive number $\Gamma$.  There is a Hermitian
polynomial $P$ of the operators $\{H_i\}_{i=1}^{n}$ of degree at
most $s$, where
\begin{align}
  s\EqDef \frac{4n\cdot 6^{b-1} \Gamma}{\sqrt{t\Delta}} 
    \qquad \text{and} \qquad P\Pi=\Pi \qquad \text{and} \qquad 
    \|P-\Pi\|\leq \frac{1}{200}\exp{\left(-\Gamma r^{b-1}\right)}
\label{eq:degeps}
\end{align}
Moreover, $P$ can be expressed as a linear combination of at most
$(2s)^{(6r)^{b-1}}$ polynomials from the set
$\mathcal{P}(s,(6r)^{b-1})$. \label{lem:nrb}
\end{lemma}
\Lem{lem:nrb} is the main contribution of this section and its
proof (below) is the conceptual heart of our method. From its
statement we can already see how the optimal tradeoff between degree
and approximation error is obtained. In particular, setting
$\Gamma=1$ and $b=O(1)$ in Lemma \ref{lem:nrb}, we see that
$r^{b-1} = \Theta\big(\frac{s^2\Delta}{n}\big)$ and so
\Eq{eq:degeps} has the desired form \eqref{eq:papprox} of
\Thm{thm:1Dpoly}. To get \Thm{thm:almostlin} we will ultimately take
$b$ to grow mildly (polylogarithmically) with $n$ which allows us to
approach linear degree.  There is also the slightly cumbersome
$\Gamma$ parameter in Lemma \ref{lem:nrb} which controls the
approximation error at the base level of the recursive construction
of the polynomial. Ultimately we take $\Gamma>1$ in the proof of
\Thm{thm:almostlin}; this ensures (for technical reasons) that the
polynomial has the claimed structure as a sum of elementary
polynomials.

\begin{proof}[Proof of Lemma \ref{lem:nrb}]
  Let us fix $t$ and $r$ satisfying \Eq{eq:n1r}. We show the claim
  by induction on $b$. 

  First consider the base case $b=1$.  In this case we have $n=t$
  and we take
  \begin{align*}
    P&\EqDef T_{\frac{\Delta}{t}, s}
      \left(\frac{1}{t} \sum_{j=1}^{t} H_j\right), &
      s &= 4\Gamma\sqrt{t/\Delta}
  \end{align*}
  which is a polynomial of degree at most $s=4\Gamma
  \sqrt{t/\Delta}=4n\Gamma/(\sqrt{t\Delta})$ as claimed. By
  construction, $P$ is Hermitian, and applying Lemma
  \ref{lem:chebyop} and using the gap property we get $P\Pi=\Pi$ and
  $\|P-\Pi\|\leq 2e^{-8\Gamma}\leq 1/200e^{-\Gamma}$ as required.
  Finally, note that $P$ is a univariate polynomial of degree at
  most $s$ in $H_S$, where $S=\{1,2,\ldots, t\}$. Thus $P$ is a
  linear combination of at most $s+1\leq 2s$ elements of
  $\mathcal{P}(s,1)$.

  Next let $b\geq 2$ and suppose the claim is true for $b-1$. Let us
  subdivide our operator labels $[n]$ into $\prt=2r-1$ overlapping
  intervals of length $n/r$ as depicted in Fig.
  \ref{fig:recursechain}. Consecutive intervals overlap in $n/(2r)$
  places (this is an integer as $t$ is even). Let us write these
  intervals as $[n]=I_1\cup I_2\cup \ldots I_\prt$. We apply the
  inductive hypothesis to obtain an approximation $P^{(j)}$ to the
  ground state projector $\Pi^{(j)}$ of each interval $I_j$ . The
  inductive hypothesis states that $P^{(j)}$ is a Hermitian
  polynomial of degree at most 
  \begin{align}
    \frac{4(n/r)\cdot 6^{b-2} \Gamma }{\sqrt{t\Delta}}
  \label{eq:bminus1}
  \end{align}
  and satisfies
  \begin{align*}
    P^{(j)}\Pi^{(j)}=\Pi^{(j)} \qquad \text{and} 
      \qquad \|P^{(j)}-\Pi^{(j)}\|
      \leq \frac{1}{200}e^{-\Gamma r^{b-2}}.
  \end{align*}
  We then define our polynomial approximation to the ground space of
  the whole chain as
  \begin{align}
    P=\widetilde{\Rob}(P^{(1)},P^{(2)},\ldots, P^{(\prt)}).
  \label{eq:P}
  \end{align}
  Since $\widetilde{\Rob}$ is a polynomial of degree at most
  $3\prt\leq 6r$, and each input is Hermitian and has degree at most
  \Eq{eq:bminus1}, $P$ is Hermitian and has degree upper bounded as
  in \Eq{eq:degeps}.  Applying Claim~\ref{clm:robground} we see that
  $P\Pi=\Pi$. Applying \Lem{lem:roblem} we get 
  \begin{align}
    \|P-\Rob(P^{(1)},P^{(2)},\ldots, P^{(\prt)})\|
      \leq \left(\frac{1}{20} e^{-\Gamma r^{b-2}}\right)^\prt.
  \label{eq:p1}
  \end{align}
  Using the merge property \eqref{eq:merge} and the triangle
  inequality we get
  \begin{align}
    \|\Pi^{(1)}\Pi^{(2)}\ldots \Pi^{(\prt)}-\Pi\|
      \leq 2\prt \exp{\left(-\frac{\sqrt{\Delta} n}{2r}\right)}.
  \label{eq:p2}
  \end{align}
  Now recall that $\Rob(P^{(1)},P^{(2)},\ldots,
  P^{(\prt)})=C^\dagger C$ where $C=\Pi^{(1)}\Pi^{(2)}\ldots
  \Pi^{(\prt)}$. Substituting in \Eq{eq:p1}, applying the triangle
  inequality, and using \Eq{eq:p2}, we arrive at
  \begin{align*}
    \|P-\Pi\|\leq \|P-C^\dagger C\| + \|C^\dagger C-\Pi\|
    \leq \left(\frac{1}{20}e^{-\Gamma r^{b-2}}\right)^{\prt}
      + 4\prt\exp{\left(-\frac{\sqrt{\Delta} n}{2r}\right)}
  \end{align*}

  To complete the proof we show that each of terms on the RHS is at
  most $1/400 \exp{(-\Gamma r^{b-1})}$. The first term is bounded in
  this way since $\prt\geq r$ and $r\geq 2$ \footnote{If $r=1$ then
  $b>1$ is the same as $b=1$ which is handled above.}. For the
  second term, we write
  \begin{align*}
    4\prt \exp{\left(-\frac{\sqrt{\Delta} n}{2r}\right)}
      = 4\prt \exp{\left(-\frac{tr^{b-2}\sqrt{\Delta}}{2}\right)}
        \leq 8r \exp{\left(-8\Gamma r^{b-1}\right)}
        \leq 1/400 \exp{(-\Gamma r^{b-1})}.
  \end{align*}
  where we used \Eq{eq:n1r} and in the last inequality we used the
  fact that $r,b\geq 2$ and $\Gamma\geq 1$.

  Finally, let us show that $P$ has the claimed structure. The
  inductive hypothesis states that each $P^{(j)}$ is a sum of at
  most $(2s/6r)^{(6r)^{b-2}}$ polynomials from the set
  $\mathcal{P}(s/6r, (6r)^{b-2})$. Using this fact and Claim
  \ref{claim:unwrap} with $m\leftarrow \prt$ we see that \Eq{eq:P}
  is a sum of at most 
  \begin{align}
    2^{5\prt} \left((2s/6r)^{(6r)^{b-2}}\right)^{3\prt}
  \label{eq:taurec}
  \end{align}
  polynomials from the set $\mathcal{P}(3\prt s/6r,
  (6r)^{b-2}3\prt)$. Using the upper bound $\prt\leq 2r$ we see that
  each of the latter elementary polynomials is in the set
  $\mathcal{P}(s, (6r)^{b-1})$, and that the number of them
  \Eq{eq:taurec} is at most $2^{10r}(2s/6r)^{(6r)^{b-1}}\leq
  (2s)^{(6r)^{b-1}}$, where we used   $2^{10r}/(6r)^{6r}\leq 1$.
\end{proof}

Let us now see how to obtain Theorems~\ref{thm:1Dpoly},
\ref{thm:almostlin} from \Lem{lem:nrb}. The proofs are along the
same lines so we handle them both below. As noted above, the key ideas are all contained in Lemma
\ref{lem:nrb} and all that is left is to choose the parameters
$t,r,b$ in a suitable manner. The analysis is somewhat tedious as
several parameters are required to be integers.

\begin{proof}[Proof of \Thm{thm:1Dpoly} and \Thm{thm:almostlin}]
  Suppose $\{H_j\}_{j=1}^{n}$ satisfy Eqs.~(\ref{eq:noncom},
  \ref{eq:locgap},\ref{eq:merge}), $\delta\in (0,1/4)$, and that $s$
  is a real number satisfying \Eq{eq:dbounds}. We shall specify an
  integer $n'\geq n$ for which a suitable polynomial approximation
  $P'$, of a certain degree $s'$, can be constructed using Lemma
  \ref{lem:nrb}. By a simple padding argument, this implies a
  polynomial approximation $P$ with degree $s'$ and the same
  approximation error for the original system of $n$
  operators\footnote{It suffices to set $H_j=0$ for all $n+1\leq
  j\leq n'$. This does not change the nullspace projector $\Pi$.
  Moreover, the new system also satisfies the gap and merge
  properties for the same $\Delta$.}. In particular, we choose
  $n'=tr^{b-1}$ where 
  \begin{align}
    b=1+\lceil \left(2\delta\right)^{-1} \rceil \qquad t
      = 2\lceil \frac{n^2}{\Delta s^2}\rceil \qquad r
      =\lceil \left(n/t\right)^{1/(b-1)}\rceil.
  \label{eq:brt}
  \end{align}
  Clearly $t,b,r$ are positive integers and $t$ is even.  Note that 
  the lower bound from \Eq{eq:dbounds} implies $t\leq n$ and
  therefore $r\leq 2(n/t)^{\frac{1}{b-1}}$. Using this fact we see
  that
  \begin{align}
    n'=tr^{b-1}\leq 2^{b-1}n.
  \label{eq:nprime}
  \end{align}
   We now show that the condition \Eq{eq:n1r} is satisfied as long as
  $\Gamma\leq n^{4\delta^2}$.
  \begin{align}
    \frac{16r\Gamma }{t\sqrt{\Delta}}
      &\leq \frac{32\Gamma}{n\sqrt{\Delta}} 
        \left(n/t\right)^{1+\frac{1}{b-1}}
      \leq  \frac{32\Gamma}{n\sqrt{\Delta}} 
        \left(\Delta s^2/2n\right)^{1+2\delta}
      \leq  \frac{32}{2^{1+2\delta}4^{2(1+2\delta)}}
        \Delta^{\delta}n^{-4\delta^2}\Gamma \leq 1 
  \label{eq:satcond}
  \end{align}
  where in the second-to-last inequality we upper bounded $s$ using
  \Eq{eq:dbounds} and in the last inequality we used the fact that
  $\Delta\leq 1$ and $\Gamma\leq n^{4\delta^2}$.  

  Let us now bound the degree $s'$ and approximation error of the
  polynomial $P'$ obtained from Lemma \ref{lem:nrb} with the choices
  \Eq{eq:brt}.  Using the fact that $t\geq 2n^2/(s^2\Delta)$ we get
  \begin{align}
    s'= \frac{4n'\cdot 6^{b-1}\Gamma }{\sqrt{t\Delta}} 
      \leq 2\sqrt{2}s\cdot (n'/n)\cdot 6^{b-1}\Gamma 
      \leq 2\sqrt{2}\cdot (12)^{b-1}s \Gamma
  \label{eq:sprime}
  \end{align}
  where we used \Eq{eq:nprime}. The approximation error satisfies
  \begin{align}
    \|P'-\Pi'\|\leq (1/200)e^{-\Gamma r^{b-1}}\leq e^{-n\Gamma/t}
      \leq e^{-\Delta s^2\Gamma/4n},
  \label{eq:approxerror}
  \end{align}
  where we used $n\leq n'$ and in the last inequality we used the
  fact that $t\leq \frac{4n^2}{\Delta s^2}$. Theorems
  \ref{thm:1Dpoly} and \ref{thm:almostlin} are obtained as special
  cases of the above.

  \Thm{thm:1Dpoly} is obtained in the special case that $\delta$ is
  a fixed constant and with the choice $\Gamma=1$. In this case we
  have $b=O(1)$ and using \Eq{eq:sprime} we see that our polynomial
  has degree $s'=O(s)$.  The approximation error \Eq{eq:approxerror}
  has the desired form since $\Gamma=1$.

  Now let us prove \Thm{thm:almostlin}. This is obtained by
  specializing to the case
  \begin{align}
    s=(1/4)\Delta^{-1/4}n^{1-\delta} \qquad \qquad 
      \delta=(\log n)^{-1/4} \qquad \qquad 
      \Gamma=n^{4\delta^2}=e^{4\sqrt{\log n}}.
  \label{eq:schoice}
  \end{align}
  Note that with these choices we have $b-1=\lceil (\log
  n)^{1/4}\rceil$. Here we have chosen $s$ at the upper limit of
  \Eq{eq:dbounds}. We also need to verify that the lower bound in
  \Eq{eq:dbounds} holds (that is, the range of allowed degrees
  \Eq{eq:dbounds} is nonempty). We see that this constraint is
  satisfied as long as $n^{2-4\delta}\Delta\geq (64)^2$. This
  follows from our assumption $n\Delta\geq C$ for some sufficiently
  large absolute constant $C$.

  \Lem{lem:nrb} then states that our polynomial $P'$ is a sum
  of at most $(2\alpha)^{\beta}$ elements of $P(\alpha,\beta)$,
  where $\alpha\EqDef s'$ and $\beta\EqDef (6r)^{b-1}$. The
  approximation error, using the first upper bound in
  \Eq{eq:approxerror}, is at most 
  \begin{align*}
    e^{-\beta \Gamma/6^{b-1}}
      = \exp{\left(-\beta e^{4\sqrt{\log n}}6^{-(b-1)}\right)}
      \leq \exp{\left(-\beta e^{\sqrt{\log(n)}}\right)}
  \end{align*}
  where we used the fact that $e^{3\sqrt{\log(n)}}\geq 6^{\lceil
  (\log n)^{1/4}\rceil}$ for $n\geq 2$.

  It remains to establish \Eq{eq:albet}. Using \Eq{eq:sprime} and
  plugging in our choices from \Eq{eq:schoice} we get 
  \begin{align*}
    \alpha\leq \frac{1}{\sqrt{2}} \Delta^{-1/4} 
      n^{1-(\log n)^{-1/4}+4(\log n)^{-1/2}} 
      (12)^{\lceil (\log n)^{1/4}\rceil}.
  \end{align*}
  Since $(12)^{\lceil (\log n)^{1/4}\rceil}=n^{O((\log n)^{-3/4})}$,
  we see from the above that for $n$ larger than some absolute
  constant we have $\alpha\leq n\Delta^{-1/4}$ as claimed (we ensure
  this by choosing $C$ sufficiently large).

  To bound $\beta$, we use the facts that $1\leq (n'/n)\leq
  n^{O\left((\log n)^{-3/4}\right)}$ (cf. Eq.~(\ref{eq:nprime})) and
  $t=\Theta(\Delta^{-1/2}n^{(\log n)^{-1/4}})$ which follows from
  our choices Eqs.~(\ref{eq:schoice},\ref{eq:brt}). Combining these
  bounds we get $\beta=6^{b-1}n'/t=\Delta^{1/2}n^{1-O((\log
  n)^{-1/4})}$ as claimed.
\end{proof}

\subsection{Application to 1D quantum spin systems}
\label{sec:1D} 

As a prototypical application of the results of the previous
section, here we specialize to the case of frustration-free
one-dimensional quantum spin systems with a local gap. 

Consider a 1D system of $n+1$ qudits of local dimension
$d\geq 2$. The Hilbert space is $\left(\mathbb{C}^d\right)^{\otimes
n+1}$ and the Hamiltonian is $H=\sum_{j=1}^{n} H_j$, where each
operator $H_j$ satisfies $0\leq H_j\leq I$ and acts nontrivially
only on qudits $j$ and $j+1$ (and as the identity on all other
qudits). The \textit{local gap} $\gamma$ is defined as the minimum
spectral gap of a subset of Hamiltonian terms 
\begin{align*}
  \gamma\EqDef \min\limits_{S\subseteq [n]} \; 
    \mathrm{gap}\big(\sum_{j\in S} H_j\big). 
\end{align*}
By definition, operators $\{H_j\}_{j=1}^{n}$ satisfy the gap
property \Eq{eq:locgap} with $\Delta=\gamma$. Below we show that the merge property is satisfied with
$\Delta=\gamma/80$ (a consequence of the ``detectability lemma"
\cite{AharonovALV08,AAV16}). Therefore we may substitute
$\Delta=\gamma/80$ in Theorems \ref{thm:1Dpoly} and
\ref{thm:almostlin} to obtain optimal approximations to the ground
state projector $\Pi$, as claimed in Theorem \ref{thm:informal1d}.

\begin{lemma}[\cite{AAV16}]
\label{lem:1Dmerge} 
  Suppose $H=\sum_{j=1}^{n}H_j$ is a 1D frustration-free qudit
  Hamiltonian with local gap $\gamma$. Then $\{H_j\}_{j=1}^{n}$
  satisfy the merge property \Eq{eq:merge} with $\Delta=\gamma/80$.
\end{lemma}

\begin{proof}[Proof sketch]
  Let $S\subseteq [n]$ be an interval partitioned as $S=ABC$. 
  Define the ``detectability lemma'' operator of the interval $S$ by
    \begin{align*}
    \DL_S \EqDef \prod_{j\in S\cap \{1,3,5,\ldots\}} 
      \Pi_j \prod_{j\in S\cap \{2,4,6,\ldots\}} \Pi_j
  \end{align*}
  where $\Pi_j$ is the projector onto the nullspace of $H_j$.
  Clearly the nullspace of $(I-\DL_S^{\dagger}\DL_S)$ is equal to
  that of $\sum_{j\in S} H_j$. Moreover, the detectability lemma, as summarized in Theorem \ref{thm:DLconv} (setting $g=2$) implies
  \begin{align*}
    \mathrm{gap}(I-\DL_S^{\dagger} \DL_S)\geq \gamma/5.
  \end{align*}
   One can show (cf. Claim~6 of \cRef{AAV16})
  \begin{align}
  \label{eq:absorb1D}
    \Pi_{AB} (\DL_S^{\dagger} \DL_S)^q\Pi_{BC}=\Pi_{AB} \Pi_{BC} 
      \qquad \text{for all integers}\qquad 0\leq q\leq \frac{|B|}{8}. 
  \end{align}
The above implies $\Pi_{AB}
  f(I-\DL_S^{\dagger}\DL_S)\Pi_{BC}=\Pi_{AB} \Pi_{BC}$ for any polynomial
  $f$ of degree at most $|B|/8$, such that $f(0)=1$. Using this and
  $\Pi_{AB}\Pi_{ABC}\Pi_{BC}=\Pi_{ABC}$ gives
  \begin{align*}
    \|\Pi_{AB}\Pi_{BC}-\Pi_{ABC}\| 
      = \|\Pi_{AB}\big(T_{\gamma/5, \frac{|B|}{8}}
        (I-\DL_S^{\dagger}\DL_S)-\Pi_{ABC}\big) \Pi_{BC}\|
      \leq 2e^{-\frac{|B|}{4}\sqrt{\gamma/5}},
  \end{align*}
  where in the last inequality we used \Lem{lem:chebyop}. 
\end{proof}

\section{2D Area law}
\label{sec:2d}

Here we consider a 2D locally gapped, frustration-free quantum spin
system along with a bipartition of the qubits into two regions. We
use the results of \Sec{sec:mergprop} to construct a polynomial
approximate ground state projector (AGSP) which has a kind of 1D
structure along the boundary of the bipartition. We show that this
AGSP has low enough error as a function of its Schmidt rank across
the bipartition, to establish the area law as stated in Theorem
\ref{thm:informalarea} using the method from
Refs.~\cite{AharonovALV08, AradLV12,AradKLV13}.

Consider a system of qudits of local dimension $d$ arranged at the
vertices of an $L\times (n+1) $ grid with $n+1$ rows and $L$
columns, as shown in \Fig{fig:set-up}. The Hilbert space is
$\left(\mathbb{C}^{d}\right)^{\otimes L(n+1)}$, and we index qudits
by their (column, row) coordinates $(i,j)\in [L]\times [n+1]$. We
consider a Hamiltonian which acts as a sum of local
projectors\footnote{This is without loss of generality. Consider a frustration-free
hamiltonian $H'=\sum_{i,j}h'_{i,j}$, where $cI \geq h'_{i,j}\geq 0$
are not projectors. Let $h_{i,j}$ be the projector onto the span of
$h'_{i,j}$, so that $c h_{i,j}\geq h'_{i,j}$. The local spectral gap
of $H_0$ is at least $\frac{1}{c}$ times the local spectral gap of
$H'$ and they have the same ground space.}
\begin{align*}
  H_0=\sum_{i=1}^{L-1}\sum_{j=1}^{n} h_{ij} \qquad h_{ij}^2=h_{ij}
\end{align*}
where $h_{ij}$ acts nontrivially only on the qudits in the set
$\{i,i+1\}\times \{j,j+1\}$. We assume that $H_0$ has a unique
ground state $\ket{\Omega}$ such that $H_0|\Omega\rangle=0$. Since $h_{ij}\geq 0$ , the latter condition is equivalent to the
frustration-free property $h_{ij}|\Omega\rangle=0$ for all $i,j$.
Our results depend on the local gap of $H_0$: 
\begin{align}
\label{eq:locgap2D}
  \gamma\EqDef \min\limits_{S\subseteq [L-1]\times [n]} \; 
    \mathrm{gap}\big(\sum_{\{i,j\}\in S} h_{ij}\big). 
      \qquad \quad[\textbf{Local gap}]
\end{align}
(recall  $\mathrm{gap}(M)$ denotes the
smallest nonzero eigenvalue of a positive semidefinite operator
$M$.) We note that for our purposes it would in fact be sufficient to
consider a local gap in which the minimization is restricted to
rectangular regions. 

We consider a bipartition of the lattice into left and right
regions, corresponding to a ``vertical cut'' between a given column $c$ and
$c+1$, as depicted in Fig.~\ref{fig:set-up}. In the
following we write $\mathrm{SR}(M)$ for the Schmidt rank of an
operator with respect to the cut. 

To bound the entanglement entropy of $|\Omega\rangle$, we use the
powerful method of approximate ground state projectors (AGSP)
developed in Refs.~\cite{Hastings07, AharonovALV08, AradLV12,
AradKLV13}. The following theorem is obtained by specializing
Corollary III.4 of \cRef{AradLV12} to the case of Hermitian $K$.

\begin{theorem}[Entanglement entropy from AGSP \cite{AradLV12}]
\label{thm:agsp} 
  Suppose $K$ is a Hermitian operator satisfying
  $K|\Omega\rangle=|\Omega\rangle$ and 
  \begin{align*}
    \|K-|\Omega\rangle\langle \Omega|\|\cdot \mathrm{SR}(K)
      \leq \frac{1}{2}.
  \end{align*}
  Then the entanglement entropy of $\ket{\Omega}$ across the cut is
  upper bounded by $O(1)\cdot\log \big( \mathrm{SR}(K)\big)$.
\end{theorem}

We use the results of \Sec{sec:mergprop} to construct a
suitable AGSP $K$. To this end we first construct a system of
operators $\{H_j\}_{j=1}^{n}$ which has the gap and merge properties
Eqs.~(\ref{eq:locgap},\ref{eq:merge}).  

We are going to focus our attention on a band of $w$ columns of
qudits centered around the cut $(c,c+1)$; see Figure
\ref{fig:Sab}(a). Here $w$ is an integer that we will choose later
to depend only on the cut length $n$. For now, suppose WLOG that the cut is not too close to the left or
right boundary of the lattice, i.e., $w\le c \le L-w$ (otherwise we
can ensure this by a padding argument). We reorganize the indices of
qudits by changing $(i,j)\rightarrow (i-(c-w/2),j)$ such that the
cut is between indices $(w/2,w/2+1)\times [n+1]$, and the region of
width $w$ is $(1,w)\times [n+1]$. Let $\Pi_{L,j}$ and $\Pi_{R,j}$
project onto the ground space of all local terms $h_{ij}$ with $i<1$
and $i\geq w$ respectively. For each row $1\leq j\leq n$ define
\begin{align}
\label{eq:truncham}
  H_j \EqDef \frac{1}{2w}\bigg(\sum_{1\leq i<w}h_{ij} 
    + (I-\Pi_{L,j}) + (I-\Pi_{R,j})\bigg),
\end{align}
see Figure \ref{fig:Sab}(a). The norm of this
operator is bounded as $\|H_{j}\|\leq \frac{1}{2w}(w+2)\leq 1$. The
following lemma is proved in Appendix \ref{sec:2Dgap} using the
detectability lemma machinery \cite{AharonovALV08, AAV16}. 
\begin{lemma}
\label{lem:2dgapmerge} 
  The operators $\{H_j\}_{j=1}^{n}$ satisfy the gap and merge
  properties Eqs.~(\ref{eq:locgap}, \ref{eq:merge}) with
  $\Delta=\Theta(\gamma/w)$.
\end{lemma}

\begin{figure}[h]
\begin{subfigure}[b]{0.4\textwidth}
\centering

\begin{tikzpicture}[xscale=0.35,yscale=0.5]

\draw [ultra thick, fill=blue!10!white] (-0.1,-0.1) rectangle (16.1,6.1);
\draw [fill=blue!30!white, dotted] (6,0) rectangle (10,6); 
\draw [fill=red!60!white, red!60!white] (0,2.8) rectangle (16,3.2); 
\node at (16.5, 3) {\footnotesize $j$};
\node at (3,3.5) {\footnotesize $\Pi_{L,j}$};
\node at (13,3.5) {\footnotesize $\Pi_{R,j}$};
\draw [dotted] (6,0) rectangle (10,6); 
\draw [ultra thick, blue] (8,-0.5) -- (8,6.5);
\draw [<->] (6,2.2) -- (10,2.2);
\node at (8.5,2.4) {\footnotesize $w$};

\end{tikzpicture}
\caption{}
\end{subfigure}
\hspace{2cm}
\begin{subfigure}[b]{0.4\textwidth}
\centering
\begin{tikzpicture}[xscale=0.35,yscale=0.5]
\draw [ultra thick, fill=blue!10!white] (-0.1,-0.1) rectangle (16.1,6.1);
\draw [fill=blue!30!white, dotted] (6,0) rectangle (10,6); 
\draw [dotted] (6,0) rectangle (10,6); 
\draw [fill=red!60!white, red!60!white] (0,2) rectangle (16,5);
\draw [thick] (0,3) -- (16, 3);
\draw [thick] (0,4) -- (16, 4);
\foreach \j in {0,...,28}
{
\draw (0.5*\j, 2) -- (2+0.5*\j,4);
\draw (0.5*\j, 5) -- (2+0.5*\j,3);
}
\foreach \j in {1,...,3}
{
\draw (0, 2+0.5*\j) -- (2-0.5*\j,4);
\draw (0, 5-0.5*\j) -- (2-0.5*\j,3);
\draw (14+0.5*\j, 2) -- (16,4-0.5*\j);
\draw (14+0.5*\j, 5) -- (16,3+0.5*\j);
}

\draw [ultra thick, blue] (8,-0.5) -- (8,6.5);
\node at (17, 2.5) {\footnotesize $A$};
\node at (17, 3.5) {\footnotesize $B$};
\node at (17, 4.5) {\footnotesize $C$};

\draw (18, 5) to [out=0, in=160] (18.5, 3.5) to [out=200, in=0] (18, 2);
\node at (19.5, 3.5) {$S$};
\end{tikzpicture}
  \caption{}
\end{subfigure}
\caption{\footnotesize (a) The vertical cut (blue line), region of width $w$ (blue), and support of the operator $H_j$ from Eq.~\eqref{eq:truncham} (red).  (b) For an interval $S\subseteq [n]$, the red region depicts the support of $H_S$. 
Given a partition $S=ABC$, the merge property asserts $\Pi_{AB}\Pi_{BC}\approxeq \Pi_S$ \label{fig:Sab} (the shading depicts $AB$ and $BC$)}
\end{figure}
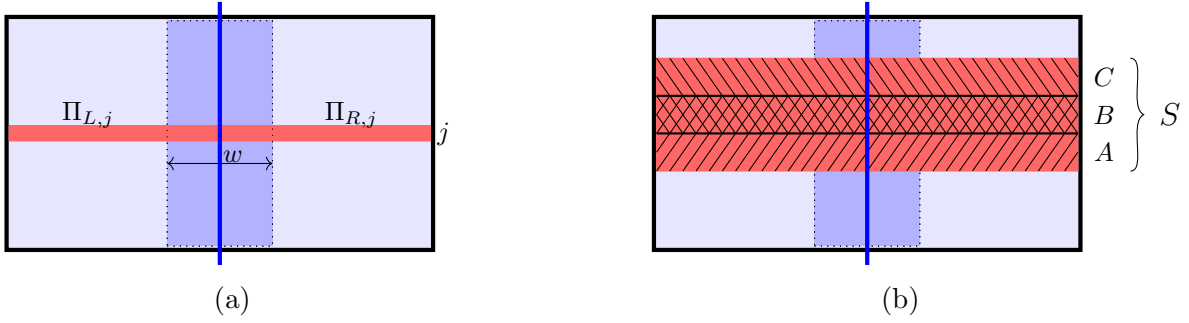

Therefore the operators $\{H_j\}_{j=1}^{n}$ satisfy the requirements
of \Thm{thm:almostlin}. The approximating polynomial from the latter
theorem is a sum of elementary polynomials from the set
$\mathcal{P}(a,b)$ from Definition \ref{def:pab}.  We shall use the
following lemma to bound the Schmidt rank of each of them. A key
feature of the bound given below is that the Schmidt rank is
amortized across the $w$ columns in the sense
that its exponential scaling with degree is $(\text{total
degree})/w$ instead of simply $(\text{total degree})$. This feature (in
a slightly different setting) was also key to the 1D area-law bound
from \cRef{AradKLV13}.
\begin{lemma}[Schmidt rank amortization]
\label{lem:srpab}
$\mathrm{SR}(Q)\leq (16a^4 d^4n)^{\frac{a}{w}+b+wn}$ for all $ Q\in \mathcal{P}(a, b)$.
\end{lemma}
The proof of \Lem{lem:srpab} is provided in
Appendix~\ref{sec:sr}. We are now in a position to prove our main result:
\begin{theorem}
\label{thm:arealawformal}
  Suppose $\gamma=\Omega(1)$ and $d=O(1)$. The entanglement entropy
  of $|\Omega\rangle$ across the cut is at most $ n^{1+O\left((\log n)^{-1/5}\right)}$.
\end{theorem}

\begin{proof}
  As a first step we apply \Thm{thm:almostlin} with the system of
  operators $\{H_j\}_{j=1}^{n}$ defined in \Eq{eq:truncham}. Lemma
  \ref{lem:2dgapmerge} states that we may set
  $\Delta=\Theta(\gamma/w)=\Theta(1/w)$. Let $P$ be the polynomial
  provided by the theorem. Our AGSP is defined as $K\EqDef P^{w^2}$,
  which is Hermitian (since $P$ is) and satisfies
  $K|\Omega\rangle=|\Omega\rangle$, since $\sum_{1\leq j\leq n} H_j$
  has unique ground state $|\Omega\rangle$. The error bound from the
  theorem gives
  \begin{align}
    \|K-|\Omega\rangle\langle \Omega\| 
      \leq e^{-\beta w^2e^{\sqrt{\log(n)}}}.
  \label{eq:shrink}
  \end{align}
  The theorem also implies that $K$ can be expressed as a linear
  combination of $(2\alpha)^{w^2\beta}$ polynomials from
  $\mathcal{P}(w^2\alpha, w^2\beta)$, where 
  \begin{align*}
    \alpha=O(nw^{1/4}) \qquad \text{and} \qquad \beta=w^{-1/2}n^{1-O\left((\log n)^{-1/4}\right)}.
  \end{align*}
  Thus, applying Lemma \ref{lem:srpab}
  \begin{align}
    \mathrm{SR}(K)\leq (2\alpha)^{w^2\beta} \max\limits_{Q\in \mathcal{P}(w^2\alpha, w^2\beta)} \mathrm{SR}(Q)\leq (16\alpha^4w^8 n d^4)^{w\alpha+2w^2\beta+wn}.
  \label{eq:srank}
  \end{align}
  As a first simplification, let us focus on the term in parentheses
  above. Our choice of $w$ below (\Eq{eq:wchoice}) satisfies $w\leq
  O(n)$ and therefore $16\alpha^4w^8 n d^4=O(\mathrm{poly}(n))$,
  where we used $d=O(1)$ and $\gamma=\Omega(1)$. Using this in
  \Eq{eq:srank} and combining with \Eq{eq:shrink} gives
  \begin{align}
    \|K-|\Omega\rangle\langle \Omega\|\cdot \mathrm{SR}(K)
      \leq \exp{\big(-\beta w^2e^{\sqrt{\log n}}
        +\big(w\alpha+\beta w^2+ wn\big)\cdot O(\log n)\big)}.
  \label{eq:ssprod1}
  \end{align}
  Taking $n$ large enough that $e^{\sqrt{\log n}}-O(\log n)\geq 1$ ,
  and substituting the values $\alpha,\beta$ gives
  \begin{align}
    \|K-|\Omega\rangle\langle \Omega\|\cdot \mathrm{SR}(K)
      &\leq \exp{\big(-\beta w^2+\big(w\alpha+wn\big)
        \cdot O(\log n)\big)}\\
    &\leq  \exp{\big(-n^{1-O\left((\log n)^{-1/4}\right)}w^{3/2}
      + O(w^{5/4}n\log n)\big)}
  \label{eq:ssprod2}
  \end{align}
  Now let us choose 
  \begin{align}
    w=\Theta(n^{(\log n)^{-1/5}}) 
  \label{eq:wchoice}
  \end{align}
  (here $1/5$ is somewhat arbitrary). With this choice, the RHS of
  \Eq{eq:ssprod2} can be made less than $1/2$ for $n$ sufficiently
  large. Applying \Thm{thm:agsp} we get that the entanglement
  entropy of $|\Omega\rangle$ across the cut is at most
  $O(1)\cdot\log(\mathrm{SR}(K))\leq n^{1+O((\log n)^{-1/5})}$.
\end{proof}


\section{Acknowledgments} 
We thank Dorit Aharonov, Fernando Brand\~{a}o, Lukasz Fidkowski, Aram Harrow, Tomotaka Kuwahara, Zeph Landau, Mehdi Soleimanifar and Umesh Vazirani for insightful discussions. AA acknowledges support through the NSF award QCIS-FF: Quantum Computing \& Information Science Faculty Fellow at Harvard University (NSF 2013303). Part of the work was done when AA was affiliated with the Simons Institute for the Theory of Computing and the Challenge Institute for Quantum Computation, where research was supported by the NSF QLCI program through grant number OMA-2016245. Part of this work was also done while AA participated in the Simons Institute for the Theory of Computing program on \emph{The Quantum Wave in Computing}. DG acknowledges the support of the Natural Sciences and Engineering Research Council of Canada, the Canadian Institute for Advanced Research, and IBM Research. IA acknowledges the support of the Israel Science Foundation (ISF) under the Individual Research Grant No.~1778/17 and the Israel Science Foundation (ISF), grant No.~2074/19.

\bibliographystyle{abbrv} 
{\footnotesize
\bibliography{references1} }


\appendix

\section{Gap and merge properties in 2D}
\label{sec:2Dgap}
Here we prove Lemma \ref{lem:2dgapmerge}, showing that the operators
$\{H_j\}_{j=1}^{n}$ satisfy the gap and merge properties
Eqs.~(\ref{eq:locgap}, \ref{eq:merge}) with
$\Delta=\Theta(\gamma/w)$. Crucial to our analysis is the following
theorem which encapsulates the detectability lemma
\cite{AharonovALV08, AAV16} and its converse \cite{Gao15}.
\begin{theorem}[Corollary 1 and Lemma 4 in \cite{AAV16}, \cite{Gao15}]
\label{thm:DLconv}
  Let $Q_1, Q_2, \ldots ,Q_m$ be a collection of Hermitian
  projectors such that each projector commutes with all but at most
  $g$ others. Define $\DL:= \prod_{i=1}^m\br{I-Q_i}$, where we fix
  some (arbitrary) ordering of terms in the product. Then
  \begin{align*}
    4\cdot\mathrm{gap}\br{\sum_{i=1}^m Q_i}
      \geq \mathrm{gap}\br{I-\DL^{\dagger}\DL}
      \geq \frac{1}{g^2+1}\cdot\mathrm{gap}\br{\sum_{i=1}^m Q_i}.
  \end{align*}
\end{theorem}

\vspace{0.1in}

\noindent{\bf Merge property:}. The proof is very similar to that of
\Lem{lem:1Dmerge}.  Fix an interval $S\subseteq [n]$ and let $\Pi_S$
be the projector onto the ground space of $H_S$.The projectors
$\{h_{ij}\}_{i\in [L-1], j\in S}$ can be partitioned into four
subsets $\cG_1, \cG_2, \cG_3, \cG_4$ such that projectors within
each subset mutually commute. Let us define $\DL_S\EqDef\prod_{k\in
[4]}\left(\prod_{h_{ij}\in \cG_k}(I - h_{i,j})\right)$. Fix a
partition $S=ABC$ (Figure \ref{fig:Sab}(b)). Analogous to
\Eq{eq:absorb1D}, we have
\begin{align*}
  \Pi_{AB} (\DL_S^{\dagger} \DL_S)^q\Pi_{BC}
    =\Pi_{AB} \Pi_{BC} \qquad \text{for all integers}
      \qquad 0\leq q\leq \frac{|B|}{8},
\end{align*}
Theorem \ref{thm:DLconv} ensures that
$\mathrm{gap}(I-\DL^{\dagger}\DL)= \Theta(\gamma)$. Thus,
\begin{align*}
  \|\Pi_{AB}\Pi_{BC}-\Pi_S\|
    =\|\Pi_{AB}\big(T_{\Theta(\gamma), \frac{|B|}{32}}
      (I-\DL_S^{\dagger}\DL_S)-\Pi_S\big) \Pi_{BC}\|
      \leq 2e^{-|B|\Theta(\sqrt{\gamma})}.
\end{align*}

\noindent{\bf Gap property:} Next we show
$\mathrm{gap}(H_S)=\Theta(\gamma/w)$, where
\begin{align*}
  H_S= \sum_{j\in S} H_j = \frac{1}{2w}
    \bigg(\sum_{1\leq i <w, j\in S}h_{ij} 
      + \sum_{j\in S}(I-\Pi_{L,j})+\sum_{j\in S}(I-\Pi_{R,j})\bigg).
\end{align*}
To this end, define
\begin{align}
\label{eq:dlprimedef}
  \DL'\EqDef \bigg(\prod_{j\in S, j=\mathrm{even}}\Pi_{L,j}\Pi_{R,j}\bigg)\bigg(\prod_{1\leq i<w, j\in S}(I-h_{ij})\bigg)\bigg(\prod_{j\in S, j=\mathrm{odd}}\Pi_{L,j}\Pi_{R,j}\bigg).
\end{align}
The nullspace $\Pi_S$ of $H_S$ is the nullspace of
$I-(\DL')^{\dagger}\DL'$. Using Theorem \ref{thm:DLconv}, we find that
$\mathrm{gap}(2wH_S)=
\Theta\br{\mathrm{gap}\br{I-(\DL')^{\dagger}\DL'}}$. To establish
the gap property, it suffices to show that
$\mathrm{gap}\br{I-(\DL')^{\dagger}\DL'}=\Theta(\gamma)$. To show this we will use the following identities, which are a consequence of $h_{ij}\Pi_{L,j}=0$
($h_{ij}\Pi_{R,j}=0$) for $i<1$ ($i\geq w$):
\begin{align*}
  \prod_{j\in S, j=\mathrm{even}}\Pi_{L,j}\Pi_{R,j}
    &=\prod_{j\in S, j=\mathrm{even}}\Pi_{L,j}\Pi_{R,j}
      \bigg(\prod_{i<1\text{ or }i
   \geq w, j\in S, j=\mathrm{even}}(I-h_{ij})\bigg),\\
    \prod_{j\in S, j=\mathrm{odd}}\Pi_{L,j}\Pi_{R,j}
  &=\bigg(\prod_{i<1\text{ or }i
   \geq w, j\in S, j=\mathrm{odd}}(I-h_{ij})\bigg)
   \prod_{j\in S, j=\mathrm{odd}}\Pi_{L,j}\Pi_{R,j}.
\end{align*}
Substituting in Eq.~(\ref{eq:dlprimedef}), we find
\begin{align*}
 \DL' &= \big(\prod_{j\in S, j=\mathrm{even}}\Pi_{L,j}\Pi_{R,j}\big)
   \DL''\big(\prod_{j\in S, j=\mathrm{odd}}
   \Pi_{L,j}\Pi_{R,j}\big),\hspace{0.5cm}\mathrm{where}\\ \DL'' &\EqDef
   \bigg(\prod_{i<1\text{ or }i\geq w, j\in S, j=\mathrm{even}}
   (I-h_{ij})\bigg)\bigg(\prod_{1\leq i<w, j\in S}(I-h_{ij})\bigg)
   \bigg(\prod_{i<1\text{ or }i\geq w, j\in S,
   j=\mathrm{odd}}(I-h_{ij})\bigg).
\end{align*}
Now $\DL''$ is a product of all the projectors $I-h_{ij}$ with $j\in S$. Applying Theorem \ref{thm:DLconv} we see that
$\mathrm{gap}\br{I-(\DL'')^{\dagger}\DL''} =
\Theta\br{\mathrm{gap}(\sum_{i, j\in S} h_{ij})}=\Theta(\gamma)$.
Furthermore, $I-(\DL')^{\dagger}\DL'$ and $I-(\DL'')^{\dagger}\DL''$
have the same nullspace $\Pi_S$. Thus, 
$\mathrm{gap}\br{I-(\DL')^{\dagger}\DL'}\geq
\mathrm{gap}\br{I-(\DL'')^{\dagger}\DL''}=\Theta(\gamma).$ This
completes the proof.

\section{Schmidt rank amortization}
\label{sec:sr}

Here we prove \Lem{lem:srpab}. Let $Q\in \mathcal{P}(a,b)$ be
given. By definition, $Q=(H_{S_1})^{j_1}(H_{S_2})^{j_2}\ldots
(H_{S_k})^{j_k}$, where $k\leq b$ and $\sum_{p=1}^{k} j_p\leq a$.
Without loss of generality we assume $k=b$ and $\sum_{p=1}^{k}
j_p=a$ (the claimed upper bound in Lemma \ref{lem:srpab} is an
increasing function of $a,b$).

Following \cRef{AradKLV13}, we introduce complex variables $Z\EqDef
\{Z_i\}_{i\in [w]}\in \mathbb{C}^{w}$ and replace $h_{ij}\leftarrow
Z_{i} h_{ij}$ for all $i\in [w]$. That is, for each $j\in [n]$, we
define
\begin{align*}
  H_{j}(Z)\EqDef \frac{1}{2w}\big(I-\Pi_{L,j}+I-\Pi_{R,j}
    + \sum_{i\in [w]}h_{ij}Z_i \big).
\end{align*}
and for $S\subseteq [n]$ let $H_S(Z)=\sum_{j\in S}H_j(Z)$.
Similarly, define
\begin{align}
  Q(Z)\EqDef (H_{S_1}(Z))^{j_1}\cdot (H_{S_2}(Z))^{j_2}\cdot 
    \ldots \cdot (H_{S_b}(Z))^{j_b}.
\label{eq:qofz}
\end{align}
We view $Q(Z)$ as a multivariate polynomial in the components of
$Z\in \mathbb{C}^{w}$ with operator-valued coefficients. We are
interested in the entanglement of $Q(Z)$ for $Z=(1,1,\ldots, 1)$
across the given vertical cut $(c,c+1)\times [n+1]$. Below we use
the notation $\mathrm{SR}_r(M)$ to denote the Schmidt rank of an
operator $M$ across a vertical cut $(r,r+1)\times [n+1]$, omitting
the subscript if $r=c$, i.e., $\mathrm{SR}(M) \EqDef
\mathrm{SR}_c(M)$. We prove the following generalization of
\Lem{lem:srpab}.
\begin{lemma}[Generalization of Lemma \ref{lem:srpab}]
\label{lem:srpabz}
  For each $Z\in \mathbb{C}^w$  we have
  \begin{align}
    \mathrm{SR}(Q(Z))\leq (16a^4 d^4n)^{\frac{a}{w}+b+wn}.
  \label{eq:qzbound}
  \end{align}
\end{lemma}

To prove this lemma, we use the following simple fact about
polynomials. For a polynomial $g(x)=\sum_{j=0}^{p} c_j x^p$ we will
use the notation $[x^\ell] g(x)\EqDef c_\ell$ to denote the
coefficient of $x^\ell$. This notation extends to multivariate
polynomials, e.g., $[x^2] (2x^2y-x^2)=2y-1$, or $[xy^2] (x+y)^3 =3$.
\begin{claim}
  Let $g(x)$ be a degree-$p$ polynomial. There exist
  $x_0,x_1,\ldots, x_{p}\in \mathbb{C}$ such that each coefficient
  $[x^c]g(x)$ can be expressed as a linear combination of
  $g(x_0),\ldots, g(x_{p})$ with complex coefficients.
  \label{claim:poly}
\end{claim}

\begin{proof}
  Take $x_j=e^{2\pi i j/(p+1)}$ for $0\leq j\leq p$. Then using the
  inverse discrete Fourier transform, we get
  $[x^c]g(x)=(p+1)^{-1}\sum_{j=0}^{p} e^{-2\pi i j c/(p+1)} g(x_j)$.
\end{proof}

We first bound the entanglement of a coefficient of a term in the
product \eqref{eq:qofz}.
\begin{lemma}
  For any $r\in [w]$, interval $S\subseteq [n]$ and integers
  $\ell\leq k\leq a$ we have 
  \begin{align*}
    \mathrm{SR}_r\big( [Z_r^\ell] (H_{S}(Z))^{k}\big)
      \leq (4a^2d^4 n)^{\ell+1}.
  \end{align*}
\label{lem:srterm}
\end{lemma}

\begin{proof}
  Let $B=\frac{1}{2w}\sum_{j\in S} h_{rj} $ and write
  \begin{align}
    H_S(Z)=A+BZ_r \qquad \qquad 
  \label{eq:opsAB}
  \end{align}
  with $A=A_{\mathrm{left}}+ A_{\mathrm{right}}$, where
  \begin{align*}
    A_{\mathrm{left}}=\frac{1}{2w} \sum_{j\in S} 
      \big( I-\Pi_{L,j}+\sum_{i=1}^{r-1} h_{ij}Z_i\big) 
        \qquad A_{\mathrm{right}}=\frac{1}{2w} 
    \sum_{j\in S}\big( I-\Pi_{R,j} + \sum_{i=r+1}^{w} h_{ij}Z_i\big).
  \end{align*}
  Using $[A_{\mathrm{left}}, A_{\mathrm{right}}]=0$ gives
  $A^p=\sum_{q=0}^{p} {p \choose q}A_{\mathrm{left}}^{q}
  A_{\mathrm{right}}^{p-q}$. Since $\mathrm{SR}_r(A_{\mathrm{left},
  \mathrm{right}})=1$ we get
  \begin{align}
    \mathrm{SR}_r(A^p)\leq p+1.
  \label{eq:ranka}
  \end{align}
  Since $h_{rj}$ acts on two qudits on each side of the cut, we have
  $\mathrm{SR}_r(h_{rj})\leq d^4$ and therefore
  \begin{align}
    \mathrm{SR}_r(B)\leq d^4|S|\leq d^4 n.
  \label{eq:rankb}
  \end{align}
  Using \Eq{eq:opsAB} we can expand
   \begin{align}
    [Z_r^\ell] (H_S(Z))^k = \sum_{x_0+\ldots+ x_\ell=k-\ell}
      A^{x_0} B A^{x_1} B\ldots B A^{x_{\ell}} 
  \label{eq:coefsum}
  \end{align}
  Using the bounds Eqs.~(\ref{eq:ranka}, \ref{eq:rankb}) and the
  fact that $x_i\leq k$ for each $i$, we see that 
  \begin{align*}
    \mathrm{SR}_r(A^{x_0} B A^{x_1} B\ldots B A^{x_{\ell}})
      \leq  (k+1)^{\ell+1}(d^{4}n)^{\ell}.
  \end{align*}
  The sum in \Eq{eq:coefsum} contains at most $(k+1)^{\ell+1}$ terms
  (since $0\leq x_i\leq k$). Hence
  \begin{align*}
    \mathrm{SR}_r( [Z_r^\ell] (H_{S}(Z))^k)
      \leq (k+1)^{2\ell+2}\cdot (d^{4}n)^{\ell}
      \leq (4a^2d^4n)^{\ell+1},
  \end{align*}
  where in the last inequality we used the fact that $k\le a$ and
  $a\ge 1$, hence $(k+1)^2\le 4a^2$.
\end{proof}

We are now in position to prove \Lem{lem:srpabz}.
\begin{proof}[Proof of \Lem{lem:srpabz}]
  Let $\mathbb{N}\EqDef\{0,1,2,\ldots\}$ and $\mathcal{K}\EqDef
  \{\alpha\in \mathbb{N}^{w}: \sum_{i=1}^{w} \alpha_i=a\}$.  We have
  \begin{align*}
    Q(Z)=\sum_{\alpha\in \mathcal{K}} Q_{\alpha} \prod_{j=1}^{w} Z_j^{\alpha_j}  \qquad Q_{\alpha} \EqDef [Z_1^{\alpha_1}Z_2^{\alpha_2}\ldots Z_w^{\alpha_w}] Q(Z).
  \end{align*}
  Since $0\leq \alpha_i\leq a$ for each $i$, we have
  $|\mathcal{K}|\leq (a+1)^{w}\leq (2a)^w$, and
73A72EA3  \begin{align}
    \mathrm{SR}(Q(Z))\leq (2a)^{w} \max_{\alpha\in \mathcal{K}} \mathrm{SR}\left(Q_{\alpha}\right)
  \label{eq:srqz}
  \end{align}
  To upper bound the RHS, let $\alpha\in \mathcal{K}$ be given. Let
  $M\EqDef \min_{i\in [w]} \alpha_i$ and let $r\in [w]$ be such that
  $\alpha_r=M$. Since $\sum_{i=1}^{w} \alpha_i=a$, we have $M\leq
  a/w$. Below we show 
  \begin{align}
    \mathrm{SR}_r(Q_{\alpha})\leq (2a)^{w+b}(4a^2d^4n)^{M+b}.
  \label{eq:srr}
  \end{align}
  Since $|r-c|\leq w/2$ and each column of the lattice contains $n$
  qudits, \Eq{eq:srr} implies 
  \begin{align*}
    \mathrm{SR}(Q_{\alpha})\leq d^{wn} \mathrm{SR}_r(Q_{\alpha})
      \leq d^{wn}(2a)^{w+b}(4a^2d^4n)^{a/w+b}
      \leq (2a)^{w+b}(4a^2d^4n)^{a/w+b+wn}
  \end{align*}
  where we used $M\leq a/w$. \Eq{eq:qzbound} follows by plugging
  into \Eq{eq:srqz} and bounding $(4a^2)^w(2a)^b\leq
  (4a^2)^{a/w+b+wn}$.

  To establish \Eq{eq:srr}, define $G(Z)\EqDef[Z_r^{M}] Q(Z)$ (a
  function of $Z_i$ with $i\neq r$) and write 
  \begin{align}
    Q_{\alpha} = [Z_1^{\alpha_1}Z_2^{\alpha_2}\ldots Z_w^{\alpha_w}] Q(Z)
      = [Z_1^{\alpha_1}\ldots Z_{r-1}^{\alpha_{r-1}} Z_{r+1}^{\alpha_{r+1}} \ldots Z_w^{\alpha_w}] G(Z).
  \end{align}
  Applying Claim~\ref{claim:poly} $w-1$ times inductively and using
  $\alpha_i\leq a$, we see that $Q_{\alpha}$ is a linear combination
  of at most $(a+1)^{w-1}$ operators $ G(X)$ with $X\in
  \mathbb{C}^{w-1}$. Therefore
  \begin{align}
    \mathrm{SR}_r(Q_{\alpha})\leq (a+1)^{w-1} \max_{X\in \mathbb{C}^{w-1}}\mathrm{SR}_r( G(X)).
  \label{eq:maxx}
  \end{align}
  Finally, using \Eq{eq:qofz} we obtain
  \begin{align*}
    G(Z)=[Z_r^{M}] Q(Z) =\sum_{\ell_1+\ell_2+\ldots +\ell_b =M} 
      [Z_r^{\ell_1}] (H_{S_1}(Z))^{j_1}) [Z_r^{\ell_2}] (H_{S_2}(Z))^{j_2})\ldots [Z_r^{\ell_b}] (H_{S_b}(Z))^{j_b}).
  \end{align*}
  Since $0\leq \ell_i\leq M$, the number of terms in the above sum
  is at most $(M+1)^b\leq (2a)^b$. Using this fact and Lemma
  \ref{lem:srterm} we get
  \begin{align*}
    \mathrm{SR}_r( G(Z))\leq \sum_{\ell_1+\ell_2+\ldots +\ell_b =M} \prod_{p=1}^{b} \mathrm{SR}_r([Z_r^{\ell_p}] (H_{S_p}(Z))^{j_p}) )\leq (2a)^{b} (4a^2d^4n)^{M+b}.
  \end{align*}
  Substituting in \Eq{eq:maxx} and bounding $(a+1)^w\leq (2a)^w$
  completes the proof of \Eq{eq:srr}.
\end{proof}

\end{document}